%% file: main.tex
\documentclass{vldb}
\usepackage{graphicx}
\usepackage{balance}  
\usepackage{mathtools}
\usepackage{pgfplots}
\usepackage{booktabs}
\usepackage{lipsum}
\usepackage{mwe}
\usepackage{color, colortbl}
\usepackage{makecell}
\usepackage{cite}

\input{localdefs.tex}

\pgfplotsset{compat=1.9}
\vldbTitle{Approximate Denial Constraints}
\vldbAuthors{Ester Livshits, Alireza Heidari, Ihab F. Ilyas, and Benny Kimelfeld}
\vldbDOI{}
\vldbVolume{}
\vldbNumber{}
\vldbYear{}

\begin{document}


\title{Approximate Denial Constraints}



%
%
%
%

\numberofauthors{4} 

\author{
%
%
\alignauthor
Ester Livshits\\
       \affaddr{Technion}\\
       \affaddr{Haifa, Israel}\\
       \email{esterliv@cs.technion.ac.il}
\alignauthor
Alireza Heidari\\
       \affaddr{University of Waterloo}\\
       \affaddr{Waterloo, ON, Canada}\\
       \email{a5heidar@uwaterloo.ca}
\alignauthor Ihab F. Ilyas\\
       \affaddr{University of Waterloo}\\
       \affaddr{Waterloo, ON, Canada}\\
       \email{ilyas@uwaterloo.ca}
\and  
\alignauthor Benny Kimelfeld\\
       \affaddr{Technion}\\
       \affaddr{Haifa, Israel}\\
       \email{bennyk@cs.technion.ac.il}}
\date{30 July 1999}

\maketitle

\begin{abstract}
The problem of mining integrity constraints from data has been extensively studied over the past two decades for commonly used types of constraints including the classic Functional Dependencies (FDs) and the more general Denial Constraints (DCs). In this paper, we investigate the problem of mining approximate DCs (i.e., DCs that are ``almost'' satisfied) from data. Considering approximate constraints allows us to discover more accurate constraints in inconsistent databases, detect rules that are generally correct but may have a few exceptions, as well as avoid overfitting and obtain more general and less contrived constraints. We introduce the algorithm $\algname{ADCMiner}$ for mining approximate DCs. An important feature of this algorithm is that it does not assume any specific definition of an approximate DC, but takes the semantics as input. 
Since there is more than one way to define an approximate DC and different definitions may produce very different results, we do not focus on one definition, but rather on a general family of approximation functions that satisfies some natural axioms defined in this paper and captures commonly used definitions of approximate constraints. We also show how our algorithm can be combined with sampling to return results with high accuracy while significantly reducing the running time.
\end{abstract}

\input{introduction.tex}

\input{related}
\input{preliminaries.tex}

\input{definitions.tex}

\input{approx.tex}

\input{algorithm.tex}
\input{sampling.tex}

\input{experiments.tex}

\input{conclusions.tex}

\bibliographystyle{abbrv}
\bibliography{main}  

\end{document}

%% file: localdefs.tex
\usepackage[noend]{algorithmic}
\usepackage{algorithm}
\usepackage{paralist}
\usepackage{xspace}
\usepackage{times}
\usepackage{xcolor}
\usepackage{multirow}

\definecolor{orchid}{RGB}{204,204,255}
\definecolor{lightorange}{RGB}{255,229,204}
\definecolor{lightgreen}{RGB}{204,255,204}
\definecolor{lightpink}{RGB}{255,204,204}
\definecolor{lightblue}{RGB}{204,229,255}
\definecolor{darkgreen}{RGB}{102,102,0}
\definecolor{darkpurple}{RGB}{153,0,76}

\newcommand{\ester}[1]{{\color{black}#1}}

\def\rhorel{\mathbin{\rho}}
\def\wrhorel{\mathbin{\widehat{\rho}}}

\def\sat{\mathbf{Sat}}
\def\evi{\mathbf{Evi}}

\def\uncov{\mathbf{uncov}}
\def\crit{\mathbf{crit}}
\def\cand{\mathbf{cand}}
\def\canCover{\mathbf{canHit}}
\def\vios{\mathbf{vios}}

\newcommand{\var}{\mathrm{var}}

\def\e#1{\emph{#1}}

\newcommand{\eat}[1]{}

\newcommand{\set}[1]{\{#1\}}

\newcommand{\algname}[1]{{\sf #1}}
\def\myrulewidth{3.20in}

\def\therule{\rule{\myrulewidth}{0.2pt}}

\newenvironment{algseries}[2]
{\begin{figure}[#1]\def\thecaption{\caption{#2}}
\vskip-0.8em\begin{tabular}{p{\myrulewidth}}\therule\end{tabular}\vskip0.2em}
{\thecaption\end{figure}}

\newenvironment{insidealg}[2]
{\normalsize\begin{insidecode}{#1}{#2}{Algorithm}}
{\end{insidecode}}

\newenvironment{insidealg2}[3]
{\normalsize\begin{insidecode2}{#1}{#2}{#3}{Algorithm}}
{\end{insidecode2}}

\newenvironment{insidesub}[2]
{\begin{insidecode}{#1}{#2}{Subroutine}}
{\end{insidecode}}

\newenvironment{insidecode}[3]
{
\begin{tabular}{p{\myrulewidth}}
\multicolumn{1}{c}{\rule{0mm}{3mm}{\bf #3} $\algname{#1}(\mbox{#2})$\vspace{-0.6em}}\\
\therule\vskip-0.8em\therule
\vspace{-1em}
\begin{algorithmic}[1]}
{\end{algorithmic}
\vskip-0.3em\therule
\end{tabular}}

\newenvironment{insidecode2}[4]
{
\begin{tabular}{p{\myrulewidth}}
\multicolumn{1}{c}{\rule{0mm}{3mm}{\bf #4} $\algname{#1}(\mbox{#2})\mbox{#3}$\vspace{-0.6em}}\\
\therule\vskip-0.8em\therule
\vspace{-1em}
\begin{algorithmic}[1]}
{\end{algorithmic}
\vskip-0.3em\therule
\end{tabular}}

\newcommand{\depset}{\mathrm{\Delta}}

\def\phi{\varphi}
\def\signature{\mathcal{S}}
\def\pspace{\mathcal{P}}

\newtheorem{theorem}{Theorem}[section]
\newtheorem{proposition}[theorem]{Proposition}
\newtheorem{lemma}[theorem]{Lemma}
\newtheorem{problem}[theorem]{Problem}
\newtheorem{definitionthm}[theorem]{Definition}
\newenvironment{definition}{\begin{definitionthm}\em}
{\end{definitionthm}}

\newenvironment{repeatresult}[2]
{\vskip0.5em\par\textsc{#1} #2.\em}
{\vskip1em}

\newtheorem{examplethm}[theorem]{Example}
\newenvironment{example}{\begin{examplethm}\em}
{\qed\end{examplethm}}

\def\att#1{\mathrm{#1}}

\newenvironment{citemize}
{\begin{compactitem}}
{\end{compactitem}}

\newenvironment{cenumerate}
{\begin{compactenum}}
  {\end{compactenum}}

\def\att#1{\mathsf{#1}}

%% file: introduction.tex
\section{Introduction}
Integrity constraints are used for stating semantic conditions that
the data in the database must comply with. Enforcing the constraints
helps to make the database a more accurate model of the real world.
Integrity constraints may be obtained by domain experts; however, this
is often an expensive task that requires expertise not only in the
domain but also in the constraint language. In the past two decades,
extensive effort has been invested in exploring the challenge of
automatically discovering constraints from the data itself, for
different types of constraints, including the classic Functional
Dependencies (FDs)~\cite{DBLP:journals/cj/HuhtalaKPT99,
  DBLP:conf/icdt/NovelliC01, DBLP:conf/edbt/LopesPL00,
  DBLP:conf/dawak/WyssGR01,
  DBLP:journals/aicom/FlachS99,DBLP:journals/tkde/LiuLLC12,DBLP:journals/pvldb/PapenbrockEMNRZ15,heidari2019holodetect},
the more general Conditional FDs
(CFDs)~\cite{DBLP:journals/pvldb/ChiangM08,DBLP:journals/tkde/FanGLX11,DBLP:conf/pkdd/RammelaereG18},
and the more general Denial Constraints
(DCs)~\cite{DBLP:journals/pvldb/ChuIP13,DBLP:journals/pvldb/BleifussKN17,DBLP:conf/dexa/PenaA18,DBLP:journals/pvldb/PenaAN19}.

In practice, databases nowadays are often inconsistent and violate the
integrity constraints that are supposed to hold. In most large
enterprises, information is obtained from imprecise and sometimes
contradicting sources (e.g., social networks, news feeds, and user
behavior data) via imprecise procedures (e.g., natural-language
processing and image
processing). 
In such cases, mining constraints that are satisfied by the entire database
will be inadequate, as they rely on the assumption that all data values
are correct. 
Hence, in this work, we consider the problem of mining \e{approximate
  constraints}, that is, constraints that are ``almost'' satisfied.
\ester{Approximate constraints are useful even for accurate datasets, since
they avoid overfitting to the current observations, and allow us to detect more general and less contrived rules, as well as rules that are generally correct but may have a
few exceptions (which is useful, for example, for the task of
detecting outliers).}
 
 
 \begin{example}\label{ex:example1}
  \ester{Consider the database of Table~\ref{table:running_example} storing information about the yearly income and tax payments of people from different states in the US. We assume that as a general rule, for a given state, it holds that a higher yearly income implies higher tax payments. However, the database does not satisfy this constraint (e.g., tuples $t_6$ and $t_7$ jointly violate the constraint, and the same holds for tuples 
  $t_{14}$ and $t_{15}$). If we consider constraints that are satisfied by the entire database, these violations require us to add additional conditions to the constraint, such as ``the constraint holds only for two people who have the same name'' or ``the constraint holds only if none of the people is called Julia and none of them lives in Illinois'', which results in very specific and complicated rules. However, we will be able to find the correct constraint if we allow for exceptions, and consider approximate constraints.}
 \end{example}
 
 \definecolor{lightgray}{gray}{0.9}
\definecolor{red}{rgb}{1.0, 0.41, 0.38}
\definecolor{green}{rgb}{0.66, 0.89, 0.63}
 
 Most of the work to date on approximate constraint discovery has focused on approximate FDs~\cite{DBLP:journals/cj/HuhtalaKPT99,DBLP:journals/cbm/CombiMSSAMP15,DBLP:conf/apweb/LiLCJY16} or CFDs~\cite{DBLP:journals/pvldb/ChiangM08,DBLP:journals/tkde/FanGLX11,DBLP:conf/pkdd/RammelaereG18}. 
 Chu et al.~\cite{DBLP:journals/pvldb/ChuIP13} and later Pena et al.~\cite{DBLP:conf/dexa/PenaA18,DBLP:journals/pvldb/PenaAN19} considered approximate DCs.  As the expressive power of (C)FDs is rather restricted, in this work, we consider the problem of mining approximate DCs (ADCs for short) from data.
This problem has not received much attention and the currently existing algorithms are $\algname{AFASTDC}$~\cite{DBLP:journals/pvldb/ChuIP13} and its improved versions $\algname{BFASTDC}$~\cite{DBLP:conf/dexa/PenaA18} and $\algname{DCFinder}$~\cite{DBLP:journals/pvldb/PenaAN19}, that we will discuss in more details in the next section. 

 
A common shortcoming of many works on approximate constraints (including the existing works on ADCs) is that the algorithms proposed for this task are often an after-thought of detecting valid exact constraints, and are usually obtained by relaxing some of the parameters of the original algorithm. Hence, existing algorithms miss opportunities to use techniques that are designed specifically for mining approximate constraints. These existing algorithms are often inefficient, since they need to examine ``all" combinations of records necessary to validate the discovered DCs.
Another drawback of existing algorithms is the fact that the approximation function is hard-wired into the algorithm. However, there are many possible definitions of approximate constraints, and different works indeed consider different definitions that produce very different results. The most common definition of approximate (C)FDs, for example, is based on the minimal number of tuples that should be removed for the (C)FD to hold~\cite{DBLP:journals/cj/HuhtalaKPT99,DBLP:journals/cbm/CombiMSSAMP15,DBLP:conf/apweb/LiLCJY16,DBLP:journals/pvldb/ChiangM08}, while the definition used for approximate DCs is based on the number of tuple pairs violating the DC~\cite{DBLP:journals/pvldb/ChuIP13,DBLP:conf/dexa/PenaA18,DBLP:journals/pvldb/PenaAN19}. It is not clear whether one of the definitions is the ``best'' one, and it may be the case that different definitions produce better results in different cases.
 
 \begin{table}[t]
  \small
    \centering
\begin{tabular}{|c||c|c|c|c|c|}
\hline
\rowcolor{lightgray} & \textbf{Name} & \textbf{State} & \textbf{Zip} & \textbf{Income} & \textbf{Tax}\\\hline
$t_1$ & Alice & NY & $11803$ & $28$K & $2.4$K\\
$t_2$ & Mark & NY & $10102$& $42$K & $4.7$K\\
$t_3$ & Bob & NY & $13914$ & $93$K & $11.8$K\\
$t_4$ & Mary & NY & $10437$ & $58$K & $6.7$K\\
$t_5$ & Alice & NY & $10437$& $26$K & $2.1$K\\
$t_6$ & Julia & WA & $98112$ & $27$K & $1.4$K\\
$t_7$ & Jimmy & WA & $98112$ & $24$K & $1.6$K\\
$t_8$ & Sam & WA & $98112$ & $49$K & $6.8$K\\
$t_9$ & Jeff & WA & $98112$ & $56$K & $7.8$K\\
$t_{10}$ & Gary & WA & $98112$ & $50$K & $7.2$K\\
$t_{11}$ & Ron & WA & $98112$ & $58$K & $8$K\\
$t_{12}$ & Jennifer & WA & $98112$ & $61$K & $8.5$K\\
$t_{13}$ & Adam & WA & $98112$ & $20$K & $1$K\\
$t_{14}$ & Tim & IL & $62078$ & $39$K & $5$K\\
$t_{15}$ & Sarah & IL & $98112$ & $54$K & $5$K\\
\hline
\end{tabular}
\caption{Running example.\label{table:running_example}}
\vspace{-1em}
\end{table}
 
 \begin{example}\label{ex:example2}
Consider again the database of Table~\ref{table:running_example} and the DC of Example~\ref{ex:example1} (i.e., $\phi_1=\forall t,t' \neg(t[\att{State}]=t'[\att{State}] \wedge t[\att{Income}]>t'[\att{Income}] \wedge t[\att{Tax}] \le t'[\att{Tax}])$). Two out of two hundred and ten pairs of tuples (i.e., $0.95\%$) violate this DC (note that $\langle t,t' \rangle$ and $\langle t',t \rangle$ are considered separately). The minimal number of tuples that should be removed from the database for the DC to hold is two (one of $t_6,t_7$ and one of $t_{14},t_{15}$); that is, $13.3\%$. Therefore, if we allow, for example, an exception rate of $5\%$, then $\phi$ will be an approximate DC according to the first definition, but it will not be an approximate DC according to the second one.  

Now, consider the DC $\phi_2=\forall t,t' \neg(t[\att{Zip}]=t'[\att{Zip}] \wedge t[\att{State}] \neq t'[\att{State}])$ (i.e., it cannot be the case that the same zip code appears for two different states). Sixteen out of two hundred and ten pairs of tuples (i.e., $7.62\%$) violate the DC (every pair of tuples that includes $t_{15}$ and one of $t_6,\dots,t_{13}$). The only tuple that needs to be removed from the database for the DC to be satisfied is $t_{15}$; thus, it is possible to remove at most $6.67\%$ of the tuples.
In this case, if the allowed exception rate is $7\%$, then $\phi_2$ is an approximate DC according to the second definition, but it is not an approximate DC according to the first one. Note that while the difference in the exception rate for these two definitions is very small here, this difference can be very significant in larger datasets.
 \end{example}

The main objective of this work is to gain a deeper understanding of ADCs and introduce a general framework for mining ADCs that takes the semantics (i.e., the approximation function) as an input. We introduce the algorithm $\algname{ADCMiner}$ for mining ADCs from data. The algorithm consists of four main components -- a predicate space generator, an evidence set constructor, an enumeration algorithm and a sampler. In summary, our main contributions in this paper are as follows:
%
\begin{citemize}
\item We formally define the problem of approximate DC mining (Section~\ref{sec:problem}), and 
  we give a formal definition of a valid approximation function (Section~\ref{sec:functions})
  that is used to define ADCs. To the  best of our knowledge, we are the first to consider approximate
  constraint discovery that is not tied to a specific  approximation
  function, but rather to a general family of approximation functions,
  that captures, but is not limited to, commonly used approximation
  functions. 
  \item We introduce an algorithm for enumerating
  ADCs that takes the approximation function as input (Section~\ref{sec:algorithm}). 
 Our algorithm is a general algorithm for enumerating \e{minimal
    approximate hitting sets} that can even be used outside the scope of
  constraint discovery.
  
    \item 
    For efficiency, we propose a \e{sampling} scheme (Section~\ref{sec:sample}), and we address two fundamental problems: \e{(1)} how to estimate the number of violations of $\phi$ in $D$ from a sample; and \e{(2)} how to use this estimate to deduce the right threshold (or approximation function) to be used when enumerating the ADCs from the sample. Sampling, while cannot be used to mine exact DCs, allows us to efficiently return highly accurate results (w.r.t.~the approximation metric) by leveraging the nature of ADCs and avoiding the space explosion, which algorithms designed for exact valid DCs suffer from.
\end{citemize}

We experimentally evaluate our proposal (Section~\ref{sec:experiments}) and show that although it subsumes previously proposed approximation frameworks, we manage to achieve better efficiency. Our experiments also show that we can achieve high precision and recall from a relatively small sample, while reducing the time by as much as $90\%$.

%% file: related.tex
\section{Related Work}
\ester{
We now discuss the relationship between our work and past work on mining DCs from data. Chu et al.~\cite{DBLP:journals/pvldb/ChuIP13} have introduced the first algorithms for mining DCs and ADCs from data ( \algname{FASTDC} and \algname{AFASTDC}, respectively). Their definition of an ADC is based on the fraction of tuple pairs violating the DC. The algorithm \algname{AFASTDC} is obtained from \algname{FASTDC} by modifying the base case of the algorithm; that is, they return a constraint if the fraction of tuple pairs violating it is smaller than some predefined threshold $\epsilon$, rather than when it is zero.
Their solution consists of two main parts. First, they generate a certain data structure, namely the evidence set, that we will formally define later on, and then they use the evidence set to generate all the (A)DCs. The first part has a very high computational cost, as it requires going over all tuple pairs in the database; hence, 
this algorithm may run for days on a database that consists of one million tuples~\cite{DBLP:journals/pvldb/ChuIP13}.

Pena et al.~\cite{DBLP:conf/dexa/PenaA18, DBLP:journals/pvldb/PenaAN19} significantly improved the running times of this part using bit-level operations, and  Position List Indexes (PLIs) that minimize the number of required tuple comparisons. Their focus was on improving the efficiency of the evidence set construction, and they did not modify the second part of the solution (that generates the ADCs) and adopted the definition of ADCs used by Chu et al.~\cite{DBLP:journals/pvldb/ChuIP13}. Our work is complementary to that of Pena et al.~\cite{DBLP:conf/dexa/PenaA18,DBLP:journals/pvldb/PenaAN19} as we focus on other aspects of ADC discovery. In particular, we do not propose a new method to construct the evidence set, but rather use the algorithm of Pena et al.~\cite{DBLP:journals/pvldb/PenaAN19} for this purpose.

Another related work is that of Bleifu{\ss} et al.~\cite{DBLP:journals/pvldb/BleifussKN17}, who introduced $\algname{Hydra}$---an algorithm that significantly improves the running times of DC discovery by incorporating sampling to invalidate candidates. However, their algorithm only works for valid exact DCs, and, as stated by the authors, it is not clear whether and how their approach can be generalized to ADCs.}

%% file: preliminaries.tex
\def\athree{a_3}
\def\aiii{a_{\mathrm{iii}}}
\def\htwo{h_2}
\def\hthree{h_3}
\def\hii{h_{\mathrm{ii}}}
\def\mtwo{m_2}
\def\mi{m_\mathrm{i}}
\def\mone{m_1}
\def\nfour{n_4}

\section{Preliminaries}\label{sec:preliminaries}

\def\sabc{\signature_{\mathrm{rl}}}
\def\dabc{\depset_{\mathrm{rl}}}
\def\stk{\signature_{2\mathrm{fd}}}
\def\dtk{\depset_{2\mathrm{fd}}}
\def\stfd{\signature_{2\mathrm{r}}}
\def\dtfd{\depset_{2\mathrm{r}}}
\def\str{\signature_{\mathrm{tr}}}
\def\dtr{\depset_{\mathrm{tr}}}
\def\rtk{R_{2\mathrm{fd}}}
\def\rabc{R_{\mathrm{rl}}}
\def\rtfd{R_{2\mathrm{r}}}
\def\rtr{R_{\mathrm{tr}}}

\begin{table}[t]
  \small
    \centering
\begin{tabular}{|c|c|}
\hline
\textbf{Notation} & \textbf{Meaning}\\\hline\hline
$S_\varphi$ & The set of predicates in the DC $\varphi$\\\hline
$\pspace_R$ & The predicate space over the relation $R$\\\hline
$\sat(t,t')$ & The set of predicates satisfied by $\langle t,t'\rangle$\\\hline
$\evi(D)$ & The evidence set of the database $D$\\
\hline
\end{tabular}
\vspace{-0.5em}
\caption{\ester{Notation table.}\label{table:notation}}
\vspace{-1em}
\end{table}

We first present some basic terminology and notation that we use
throughout the paper.

By $R(A_1,\dots,A_k)$ we denote a relation symbol $R$
with the attributes $A_1,\dots, A_k$.
 A \e{database} $D$ over a relation $R(A_1,\dots,A_k)$ is a finite set of \e{tuples} $(c_1,\dots,c_k)$ where each
$c_i$ is a constant. We denote by $t[A_i]$ the value of tuple $t$ in attribute $A_i$. 
%

\def\rhorel{\mathbin{\rho}}

A \e{denial constraint} (DC for short) is an expression of the form $\forall x \neg(\omega(x)\wedge\psi(x))$, where $x$ is a sequence of variables, $\omega(x)$ is a conjunction of atomic formulas and $\psi(x)$ is a conjunction of comparisons between two variables in $x$. \ester{Following previous works on the problem of mining DCs~\cite{DBLP:journals/pvldb/ChuIP13,DBLP:journals/pvldb/BleifussKN17,DBLP:conf/dexa/PenaA18,DBLP:journals/pvldb/PenaAN19}}, we limit ourselves to DCs where $\omega(x)$ is a conjunction of precisely two atomic formulas over the same relation and the comparison operators are $\mathbb{B}=\set{=,\neq,>,<,\ge,\le}$. 

Let $R$ be a relation and let $D$ be database over $R$. The \e{predicate space} $\pspace_R$ from which DCs can be formed consists of predicates of the form $t[A]\rhorel t'[B]$, where $A$ and $B$ are attributes of $R$, and $\rho$ is a comparison operator from $\mathbb{B}$.  Throughout the paper, we will use the following notation for DCs: $\forall t,t' \neg(P_1,\dots,P_m)$, where each $P_i$ is a predicate from $\pspace_R$.
The \e{complement} of a predicate $t[A] \rhorel t'[B]$ is the predicate $\widehat{P}=t[A] \wrhorel t'[B]$, where $\widehat{\rho}$ is the complement operator of $\rho$ (e.g., the complement operator of $>$ is $\le$). The complement of a set $S=\set{P_1,\dots,P_m}$ of predicates is the set of predicates $\set{\widehat{P_1},\dots,\widehat{P_m}}$. We denote this set by $\widehat{S}$.

For a pair $\langle t,t'\rangle$ of tuples in a database $D$ over $R$, we denote by $\sat(t,t')$ the set of all predicates in $\pspace_R$ satisfied by $\langle t,t'\rangle$. We denote by $\evi(D)$ the set $\set{\sat(t,t')\mid t,t'\in D}$, which we refer to as the \e{evidence set}~\cite{DBLP:journals/pvldb/ChuIP13}. \ester{Throughout the paper we assume the bag semantics for $\evi(D)$, as the number of occurrences of each set in $\evi(D)$ is important. In practice, we store every set in $\evi(D)$ once, along with its number of occurrences.} We identify a DC $\phi$ with the set $S_\phi$ of
its predicates. A DC states that its predicates cannot be satisfied all at the same time. That is, a DC $\phi$ is satisfied by a tuple pair $\langle t,t'\rangle$ if at least one of the predicates $P\in S_\phi$ does not hold for $\langle t,t'\rangle$, or, equivalently, $\widehat{P}\in \sat(t,t')$.
A DC $\phi$ is \e{satisfied} by a database $D$ (denoted by $D\models\phi$) if it is satisfied by all pairs of tuples , and  \e{violated} otherwise. If a DC $\phi$ is satisfied by a database $D$, we say that it is a \e{valid} DC w.r.t.~$D$. 

\begin{example}
Table~\ref{table:pspace} contains a subset of the predicate space $\pspace_R$ over the relation of our running example. We use the operations in $\set{<,\le,>,\ge}$ only for numeric attributes,  and  we only allow comparisons among attributes of the same type (i.e., two numeric or string attributes). For example, the predicate $t[\att{Name}]=t[\att{Income}]$ will not appear in $\pspace_R$. Among the predicates of Table~\ref{table:pspace}, the predicate set $\sat(t_2,t_5)$ of the tuples $t_2$ and $t_5$ of our running example will contain the predicates $t[\att{Name}]\neq t'[\att{Name}]$,
$t[\att{Income}]>t'[\att{Income}]$,
$t[\att{Income}]\ge t'[\att{Income}]$, and
$t[\att{Income}]>t'[\att{Tax}]$. The set $\sat(t_5,t_2)$ will also
contain the first two predicates, but it will not contain the other two
predicates; instead, $t[\att{Income}]<t'[\att{Income}]$ and
$t[\att{Income}]\le t'[\att{Income}]$ will appear in the set.
\end{example}

\ester{In principle, our solution could be extended to more general DCs. For example, we could relax the limitation on the number of atomic formulas, which will affect mainly the size of $\evi(D)$ (i.e., if we allow for $k$ atomic formulas, then $\evi(D)$ will contain a set $\sat(t_1,\dots,t_k)$ for each sequence $t_1,\dots,t_k$ of tuples in $D$, and each such set will consist of more predicates, as $t_1[A]=t_2[A]$ is different than $t_2[A]=t_3[A]$). We could also consider other types of predicates, such as $t[A]\rhorel (k\times t'[B])$, which will increase the size of the predicate space. However, such extensions will have a significant impact on the running times, and the trade-off between more general constraints and lower running times has to be taken into account. When we focus on the DCs considered in this paper, we are already able to discover many constraints that cannot be discovered using FD discovery methods. In our experiments, about $70\%$ of the discovered constraints cannot be expressed as FDs.}

\begin{table}[t]
  \small
    \centering
\begin{tabular}{|c c c c|}
\hline
$t[\att{Name}]=t'[\att{Name}]$ & & & $t[\att{Name}]\neq t'[\att{Name}]$\\
$t[\att{Income}]=t'[\att{Income}]$ & & & $t[\att{Income}]\neq t'[\att{Income}]$\\
$t[\att{Income}]>t'[\att{Income}]$ & & & $t[\att{Income}]\ge t'[\att{Income}]$\\
$t[\att{Income}]<t'[\att{Income}]$ & & & $t[\att{Income}]\le t'[\att{Income}]$\\
$t[\att{Income}]>t'[\att{Tax}]$ & & & $t[\att{Income}]\ge t'[\att{Tax}]$\\
$t[\att{Income}]<t'[\att{Tax}]$ & & & $t[\att{Income}]\le t'[\att{Tax}]$\\
\hline
\end{tabular}
\vspace{-0.5em}
\caption{A sample of the predicate space of our example.\label{table:pspace}}
\vspace{-1em}
\end{table}

%% file: definitions.tex
\section{Problem and Solution Overview}\label{sec:problem}
In this section, we formally define the problem that we study in the paper and give an overview of our solution.

\subsection{Problem Definition}

We start by defining a \e{valid approximation function}.
Let $D$ be a database, and let $\phi$ be a DC.
Let $f$ be a function $f:(D,S_\phi)\rightarrow [0,1]$. We now define two properties of such a function $f$, namely, \e{Monotonicity} and \e{Indifference to Redundancy}.

\begin{definition}[Monotonicity]
A function $f:(D,S_\phi)\rightarrow [0,1]$ is monotonic if it holds that $f(D,S_{\phi})\le f(D,S_{\phi'})$ whenever $S_{\phi}\subset S_{\phi'}$.\qed
\end{definition}

Intuitively, monotonicity ensures that the more predicates a DC contains, the higher its score is, as the number of tuple pairs that satisfy the DC can only increase.
Monotonicity allows us to consider only minimal ADCs (i.e., ADCs that do not strictly contain any ADC), as it assures that whenever $\phi$ is an ADC, every $\phi'$ such that $S_\phi\subset S_{\phi'}$ is also an ADC. Hence, when returning only minimal ADCs $\phi$, we also implicitly provide the user with information on any $\phi'$ that can be obtained from $\phi$ by adding more predicates. For non-monotonic functions, on the other hand, it may be the case, for example, that for $\phi,\phi'$ and $\phi''$ such that $S_\phi\subset S_\phi'\subset S_\phi''$, the DCs $\phi$ and $\phi''$ are ADCs, while $\phi'$ is not. Thus, returning only $\phi$ will result is the loss of valuable information (that is, the fact that $\phi'$ is not an ADC), and it will be necessary to go over the entire space of possible ADCs to make sure that we return all of them.


\begin{definition}[Indifference to Redundancy]
A function $f:(D,S_\phi)\rightarrow [0,1]$ is indifferent to redundancy if we have that $f(D,S_\phi)=f(D,S_\phi')$ whenever $S_\phi\subset S_{\phi'}$ and $\set{\langle t,t' \rangle\mid t,t'\in D,\set{t,t'}\models\phi}=\set{\langle t,t' \rangle\mid t,t'\in D,\set{t,t'}\models\phi'}$.\qed
\end{definition}

A function $f$ is indifferent to redundancy if adding more predicates to a DC $\phi$ without affecting the coverage, does not affect the score; that is, if two DCs $\phi$ and $\phi'$ such that $S_\phi\subset S_{\phi'}$ are satisfied by the exact same tuple pairs, then $f$ gives them the same score. While our algorithm for enumerating minimal ADCs could work for functions that do not satisfy indifference to redundancy, having this property allows us to significantly increase the algorithm efficiency by pruning the search tree early, as we explain in Section~\ref{sec:algorithm}. 

We now define valid approximation functions.

\begin{definition}[Valid Approximation Function]
A function $f:(D,S_\phi)\rightarrow [0,1]$ is a valid approximation function if it satisfies monotonicity and indifference to redundancy.\qed
\end{definition}

In the next section, we will show that this definition is quite general and captures commonly used approximation functions. Next, we give the formal definition of a minimal ADC.

\begin{definition}[Approximate Denial Constraint]
Let $D$ be a database, let $f$ be a valid approximation function, and let $\epsilon\ge 0$. Then, a DC $\phi$ is a minimal ADC if:
\begin{cenumerate}
    \item $1-f(D,S_\phi)\le\epsilon$, \e{and}
    \item no DC $\phi'$ s.t.~$S_{\phi'}\subset S_\phi$ satisfies $1-f(D,S_{\phi'})\le\epsilon$.\qed
\end{cenumerate}
\end{definition}

The intuition behind using valid approximation functions (i.e., combining the two properties) when considering ADCs is illustrated in the following example.

\begin{example}
Consider the following DCs:
\begin{align*}
\phi=&\forall t,t' \neg(t[A]<t'[A]\wedge t[A]\le t'[A])\\
  \phi'=&\forall t,t' \neg(t[A]<t'[A])
          \end{align*}
The DC $\phi'$ is satisfied by the exact same pairs of tuples from $D$ as $\phi$, since whenever a tuple pair satisfies the predicate $t[A]<t'[A]$ it also satisfies $t[A]\le t'[A]$. Intuitively, the DC $\phi'$ is minimal, while $\phi$ is not minimal, as there is no benefit in adding the predicate $t[A]\le t'[A]$ to the DC. For a monotonic function $f$, it will hold that $f(D,S_{\phi'})\le f(D,S_\phi)$; however, it may be the case that $1-f(D,S_\phi)\le\epsilon$, while $1-f(D,S_\phi')>\epsilon$, in which case we will return $\phi$ and not $\phi'$. The existence of the second property (i.e., indifference to redundancy) resolves this problem since, as aforementioned, the same pairs of tuples satisfy both DCs; thus, we have that $f(D,S_\phi)=f(D,S_{\phi'})$ and we will either return $\phi'$ (if $1-f(D,S_{\phi'})\le\epsilon$) or none of the DCs.
\end{example}


\ester{Finally, we define the problem that we study in this paper. 
\begin{problem}[ADC Mining Problem]
For a database $D$, an approximation function $f$, and a threshold $\epsilon\ge 0$,  generate all the nontrivial minimal ADCs for $D$ w.r.t.~$f$ and $\epsilon$.
\end{problem}
Since generating ADCs from the entire database may be very time consuming for large databases, we also consider the problem of discovering ADCs from a sample.}

\subsection{ADCMiner}

Our algorithm, $\algname{ADCMiner}$ is depicted in Figure~\ref{alg:adcminer}. The input to the algorithm consists of a database $D$ over a relation $R$, a valid approximation function $f$, and an approximation threshold $\epsilon\ge 0$. The following are the four main components of the algorithm.
\begin{cenumerate}
\item A {\it predicate space generator}, which builds the predicate space $\pspace_R$ for the given relation $R$. We use the algorithm of Chu et al.~\cite{DBLP:journals/pvldb/ChuIP13} for this task. The predicates in $\pspace_R$ may compare the same attribute in two different tuples (i.e., $t[A]\rhorel t'[A]$), two different attributes in the same tuple (i.e., $t[A]\rhorel t[B]$), or two different attributes in two tuples (i.e., $t[A]\rhorel t'[B]$). We allow comparing two attributes only if they have at least $30\%$ common values as in~\cite{DBLP:journals/pvldb/ChuIP13,DBLP:journals/pvldb/PenaAN19}. \ester{In principle, it is possible to compare attributes with less than 30\% common values; however, relaxing this requirement may also significantly increase the number of unuseful predicates (like $t_1[Age]\neq t_2[Zip]$). The experiments conducted by Chu et al.~\cite{DBLP:journals/pvldb/ChuIP13} have shown that requiring at least 30\% common values allows us to identify many of the comparable attributes, while avoiding a significant increase in the number of meaningless predicates.}

\item A {\it sampler}, which  draws a random sample $J$ of tuples from $D$. We provide a theoretical analysis of mining ADCs from a sample in Section~\ref{sec:sample} and experimentally evaluate the accuracy of the results obtained from a sample in Section~\ref{sec:experiments}.

\item An {\it evidence set generator}, which builds the evidence set from the sample $J$. In this paper, we use an existing algorithm for constructing the evidence set~\cite{DBLP:journals/pvldb/PenaAN19}.

\item An {\it enumeration algorithm}, which takes as input the sample $J$, the evidence set $\evi(J)$, the approximation function $f$ and the approximation threshold $\epsilon$ and enumerates all the minimal ADCs of $J$ w.r.t.~$f$ and $\epsilon$ (cf. Section~\ref{sec:algorithm}).
\end{cenumerate}

Note that ADCs allow exceptions by definition, and can be seen as DCs obtained from a sample, where the sample consists of the subset of tuples that jointly satisfy the DC. Hence, we are able to obtain good results from a sample, instead of using the whole database $D$. Our experimental evaluation shows that using a sample of $30\%-40\%$ of the tuples, we consistently obtain results with a high $F_1$ score (compared to  mining the whole database), while reducing the running time by as much as 90\%.

{
\begin{algseries}{t}{\label{alg:adcminer} An algorithm for discovering ADCs.}
\begin{insidealg}{ADCMiner}{$R$, $D$, $f$, $\epsilon$}
\STATE $\pspace_R=\algname{GeneratePSpace(R)}$
\STATE $J=\algname{Sample}(D)$
\STATE $\evi(J)=\algname{ConstructEvidence}(J)$
\STATE $\algname{ADCEnum}(J,\evi(J),\pspace_R, f,\epsilon)$
\end{insidealg}
\end{algseries}
}

%% file: approx.tex
\section{Approximation Functions}\label{sec:functions}
In this section, we discuss three specific valid approximation functions. Kivinen et al.~\cite{Kivinen1992} introduced three definitions of approximate FDs, based on three different measures, which can be easily generalized to DCs. We start by discussing each one of these measures and the corresponding approximation functions. 

\ester{
Let $D$ be a database and let $\varphi$ be a DC. The first measure proposed by Kivinen et al.~\cite{Kivinen1992} (denoted by $g_1$) is based on the proportion of tuple pairs violating the constraint. Formally, we define the following approximation function based on this measure:
\[f_1(D,S_\phi)=\left|\set{\langle t,t' \rangle\mid t,t'\in D, \{t,t'\}\models\phi}\right|/|D|^2\]
Note that in our definition we count the pairs satisfying the constraint; hence, we have that $g_1(D,\varphi)=1-f(D,S_\phi)$.
Intuitively, $f_1(D,S_\phi)$ is the probability to select a satisfying tuple pair among all pairs, assuming a uniform distribution of the violations.
This measure has been used in~\cite{DBLP:journals/pvldb/ChuIP13} and ~\cite{DBLP:conf/dexa/PenaA18,DBLP:journals/pvldb/PenaAN19} to define ADCs.
}

\ester{
The second measure in~\cite{Kivinen1992}, denoted by $g_2$, is based on the proportion of ``problematic'' tuples (i.e., tuples that are involved in a violation of the constraint). Here, we define the following approximation function:
\[f_2(D,S_\phi)=\left|\set{t\mid t\in D, \not\exists t'\in D, \set{t,t'}\not\models \phi}\right|/|D|\]
}
\ester{
Again, we have that $g_2(D,\phi)=1-f_2(D,S_\phi)$. If we consider an inconsistent database $D$, it may be the case that only one tuple contains errors, but every pair of tuples that includes this tuple violates the DC $\phi$. In this case, it holds that $f_2(D,S_\phi)=0$, as all the tuples appear in one violating pair. However, if we just remove this one tuple, the DC will hold. Thus, this measure may be too sensitive, and the last measure ($g_3$) proposed by Kivinen et al.~\cite{Kivinen1992}, that is based on the minimal number of tuples to remove from the database for the constraint to hold,
seems to be a better fit in this case. Hence, we introduce the following approximation function.
\[f_3(D,S_\phi)=\max_{D'}\set{|D'|\mid D'\subseteq D, D'\models \phi}/|D|\]
That is, the value $f_3(D,S_\phi)$ (or, equivalently, $1-g_3(D,\phi)$) is the size of a \e{cardinality repair}~\cite{DBLP:conf/icdt/LopatenkoB07} of $D$ (i.e., the largest subinstance of $D$ among all those satisfying the DC). The subinstance $D'$ considered in this function can also be seen as a Most Probable Database~\cite{Gribkoff2014TheMP} in the framework of tuple independent probabilistic databases. This approximation function has been used in many works on approximate (C)FDs~\cite{DBLP:journals/cj/HuhtalaKPT99,DBLP:journals/cbm/CombiMSSAMP15,DBLP:conf/apweb/LiLCJY16,DBLP:journals/pvldb/ChiangM08}.}


%
%


We now prove that the functions $f_1$, $f_2$ and $f_3$ satisfy both monotonicity and indifference to redundancy.

\begin{proposition}
The functions $f_1,f_2$, and $f_3$ are monotonic.
\end{proposition}
\begin{proof}
The denominator does not depend on $\phi$ in any of the three functions; hence, monotonicity only depends on the numerator.
Clearly, the function $f_1$ is monotonic, as adding more predicates to $\phi$ can only increase the number of tuple pairs that satisfy the DC. For that same reason, the number of tuples $t\in D$ for which we have that for every $t'\in D$ both $\langle t,t'\rangle$ and $\langle t',t\rangle$ satisfy $\phi$ can only increase, and the function $f_2$ is also monotonic. Finally, we prove that $f_3$ is monotonic. Let $D'$ be a subinstance of $D$ such that $D'\models\phi$ and there is no other subinstance $D''$ of $D$ that also satisfies this property such that $|D''|>|D'|$. Clearly, for each $\phi'$ such that $S_\phi\subseteq S_{\phi'}$ it holds that $D'\models\phi'$ as well. Thus, $D'$ also satisfies the condition in the numerator of $f_3$ for $\phi'$ (although $D'$ is not necessarily maximal in this case), and the value $f_3(D,S_{\phi'})$ cannot be lower than $f_3(D,S_\phi)$.
\end{proof}

\begin{proposition}
The functions $f_1,f_2$, and $f_3$ are indifferent to redundancy.
\end{proposition}
\begin{proof}
The fact that this property is satisfied by $f_1$ and $f_2$ is rather straightforward. If the same tuple pairs satisfy both $\phi$ and $\phi'$, then clearly the function $f_1$ that counts such pairs assigns the same value to both DCs. This also implies that the tuples involved in violations of both DCs are exactly the same, which means that $f_2(D,S_\phi)=f_2(D,S_{\phi'})$ as well. To prove indifference to redundancy for $f_3$, we will show that every subinstance $D'$ of $D$ satisfies $\phi$ if and only it satisfies $\phi'$. This holds since every subinstance $D'$ satisfying one of these DCs does not contain any pair of tuples from $D$ that jointly violate the DC, and since the exact same pairs of tuples from $D$ violate both DCs, it means that it does not contain any tuple pair violating the other DC.
\end{proof}

We also prove the following result regarding the relationships between the functions $f_2$, $f_3$ and the function $f_1$. As will be seen in the next section, throughout the algorithm we always keep track of the sets in $\evi(D)$ that have an empty intersection with $\widehat{S_\phi}$; hence, we can compute the function $f_1$ faster than computing $f_2$ or $f_3$. The next proposition allows us to reduce the number of times we are required to compute $f_2$ or $f_3$ using the function $f_1$.

\begin{proposition}\label{prop:f3-f1}
Let $D$ be a database, $\phi$ a DC, and $\epsilon\ge 0$. For $i\in\set{2,3}$, if $1-f_i(D,S_\phi)\le\epsilon$ then $1-f_1(D,S_\phi)\le 2\epsilon$.
\end{proposition}
\begin{proof}
The evidence set $\evi(D)$ contains \ester{$2(|D|-1)$} sets for every tuple $t\in D$ (two sets, $\sat(t,t')$ and $\sat(t',t)$, for every tuple $t'\in D$). If $1-f_2(D,S_\phi)\le\epsilon$, then at most $\epsilon|D|$ tuples appear in a violating pair. Thus, the number of violating pairs is at most \ester{$2\epsilon|D|(|D|-1)$}, which is exactly $2\epsilon$ of the tuple pairs. We conclude that $1-f_1(D,S_\phi)\le 2\epsilon$. As for the function $f_3$, when we remove a tuple from $D$, we remove \ester{$2(|D|-1)$} sets from $\evi(D)$. If $1-f_3(D,S_\phi)\le\epsilon$, then there is a subinstance $D'$ of $D$ that is obtained by removing at most $\epsilon|D|$ tuples from $D$ such that $D'\models\phi$. This observation implies that $\evi(D')$ contains every set in $\evi(D)$ except for at most \ester{$2\epsilon|D|(|D|-1)$} sets. Since $D'$ satisfies $\phi$, at most \ester{$2\epsilon|D|(|D|-1)$} pairs violate $\phi$, which is at most $2\epsilon$ of the tuple pairs, and again we have that $1-f_1(D,S_\phi)\le 2\epsilon$.
\end{proof}

Finally, we discuss the computational complexity of the three functions. Unlike the functions $f_1$ and $f_2$ that can be computed in polynomial time for both FDs and DCs, the function $f_3$ can be computed in polynomial time for FDs~\cite{DBLP:conf/pods/LivshitsKR18}, but not for DCs. 
Livshits et al.~\cite{DBLP:journals/corr/abs-1904-06492} have shown that this problem is NP-hard even when considering simple DCs over a single relation symbol (e.g., the DC $\forall t,t'\neg(t[A]\neq t'[B])$). Hence, we cannot efficiently compute $f_3$. However, there is a simple reduction from the problem of computing $1-f_3(D,S_\varphi)$ to the minimum vertex cover problem (where the goal is to find a minimal set of vertices that intersects with all the edges), based on the concept of a \e{conflict graph}, 
in which vertices represent tuples and edges represent violations. Since vertex cover is 2-approximable in polynomial
time~\cite{DBLP:journals/jal/Bar-YehudaE81}, this is also the case for our problem. Thus, to generate ADCs w.r.t.~$f_3$ we could use the $2$-approximation algorithm with the threshold $2\epsilon$. Note that we will return all ADCs, but we may also return some DCs for which it holds that $1-f_3(D,S_\phi)\le 2\epsilon$ but $1-f_3(D,S_\phi)>\epsilon$.

{
\begin{algseries}{t}{\label{alg:f3_appx} A greedy algorithm replacing $f_3$.}
\begin{insidealg}{\ester{GreedyF3}}{$D,S_\varphi,\vios,\epsilon$}
\STATE $(T,v)=\algname{SortTuples(D,S_\varphi,\vios)}$
\STATE $u=\left|S\in\evi(D)\mid S\cap S_\varphi=\emptyset\right|$
\STATE $c=0$, $R=\emptyset$
\WHILE{$c<u$}
\STATE let $t$ be the first tuple in $T$
\STATE $c=c+v(t)$
\STATE remove $t$ from $T$ and add it to $R$
\ENDWHILE
\RETURN $\left(|R|/|D|\le\epsilon\right)$
\end{insidealg}
\begin{insidesub}{SortTuples}{$D,S_\varphi,\vios$}
\STATE $v(t)=0$ for all $t\in D$
\FORALL{$S\in\evi(D)$ such that $S\cap S_\varphi=\emptyset$}
\FORALL{$t\in\vios[S]$}
\STATE $v(t)=v(t)+\vios[S][t]$
\ENDFOR
\ENDFOR
\RETURN $(\mbox{tuples of } D \mbox{ in descending order of } v(t), v(t))$
\end{insidesub}
\end{algseries}
}

\ester{In practice, the 2-approximation algorithms for minimum vertex cover assume an explicit representation of the graph. In our case, this requires storing, for every set $S$ in $\evi(D)$, all pairs $\langle t,t'\rangle $ of tuples such that $\sat(t,t')=S$. As the number of tuple pairs is quadratic in the size of the database, storing this information with reasonable memory usage is infeasible for large databases. Hence, in our experimental evaluation, we implement a greedy algorithm (depicted in Figure~\ref{alg:f3_appx}) instead. This greedy algorithm is inspired by the greedy $O(\log n)$-approximation algorithm for minimum vertex cover, that, in each iteration, selects a vertex that is adjacent to the maximal number of uncovered edges, and then marks each one of these edges as covered. However, our algorithm does not require an explicit representation of the graph; hence, we do not know which edges are covered. While we do not provide any theoretical guarantees on the result of this algorithm, our experimental evaluation shows that using this algorithm we often obtain more accurate results than the ones obtained using the function $f_2$.

In the algorithm, we sort the tuples in descending order according to the number of violations they participate in. For that, we use the data structure $\vios$ that stores, for every set $S\in\evi(D)$ and tuple $t\in D$, the number of violations of type $S$ that $t$ is involved in (that is, the number of tuple pairs $\langle t_1,t_2\rangle$ such that $\sat(t_1,t_2)=S$ and either $t_1=t$ or $t_2=t$). Then, we start selecting these tuples, one by one, while recording the change to the number of violations covered by the selected tuples. That is, with every tuple that we select, we add the number of violations it participates in to the number of covered violations $c$. We stop this process when the number of covered violations $c$ is at least the number of total violations $u$. The number of covered violations can be higher than the number of total violations, as if two tuples $t,t'$ jointly violate the DC and are both added to the result, we count this violation twice. Finally, we return the DC if the ratio between the number of tuples in the result and the total number of tuples is lower than the threshold. 

The most time consuming part of the algorithm is the subroutine \algname{SortTuples}; hence, the time complexity is $O(|D|\cdot n)$ where $n$ is the number of \e{distinct} sets in $\evi(D)$ (recall that we treat $\evi(D)$ as a bag), and the space complexity, which depends on the size of $\vios$, is the same. In all of our experiments, the number of distinct sets in $\evi(D)$ is orders of magnitude smaller than the number of tuple pairs; hence, storing this data structure requires significantly less space than storing data for every pair of tuples.}

%% file: algorithm.tex
\section{Enumeration Algorithm}\label{sec:algorithm}

In this section, we introduce an algorithm for enumerating minimal ADCs. Following Chu et al.~\cite{DBLP:journals/pvldb/ChuIP13}, we reduce our problem to that of enumerating \e{minimal approximate hitting sets}. The hitting set problem is the following: given a finite set $K$ and a family $M$ of subsets of $K$, find all subsets of $K$ that intersect every one of the subsets in $M$. A subset $F$ is a \e{minimal} hitting set if no proper subset of $F$ is a hitting set. As mentioned in the preliminaries, a pair $\langle t,t'\rangle$ of tuples satisfies a DC $\phi$ if $\widehat{P}\in\sat(t,t')$ for some $P\in S_\phi$. Hence, it is rather straightforward that $\phi$ is a valid DC if $\widehat{S_\phi}$ is a hitting set of $\evi(D)$. Note that the other direction does not necessarily hold, as a hitting set may not correspond to a nontrivial DC. For example, the set $\set{t[A]=t'[A],t[A]\neq t'[A]}$ is clearly a hitting set of $\evi(D)$, but the corresponding DC is trivial. Hence, the reduction is essentially to the hitting set problem with restrictions rather than the general hitting set problem.

Although the complexity of enumerating minimal hitting sets or, equivalently, hypergraph transversals 
is still an open problem (after decades of research), many algorithms have been proposed for this task (see~\cite{DBLP:journals/siamdm/Gainer-DewarV17} for a survey). Yet, to the best of our knowledge, the problem of enumerating minimal \e{approximate} hitting sets has not received much attention. \ester{Here, we refer to a set $F\subseteq K$ that satisfies $1-f(M,F)\le\epsilon$ for a given valid approximation function $f$ and a threshold $\epsilon$ as an approximate hitting set.} Researches typically refer to one of two problems as computing approximate hitting sets: \e{(1)} enumerating hitting sets, but not necessarily all of them (and not necessarily minimal)~\cite{DBLP:conf/sara/AbreuG09,DBLP:conf/musepat/CardosoA13,DBLP:conf/cla/NourineQT15}, \e{and (2)} computing an approximate hitting set of minimum cardinality~\cite{DBLP:journals/ijar/VinterboO00,DBLP:conf/soda/ChandrasekaranKMV11,DBLP:journals/dam/BusMR18}. However, we focus on the problem of generating minimal approximate hitting sets for a given approximation function. Hence, we devise an algorithm for enumerating minimal approximate hitting sets, building upon an algorithm for enumerating minimal hitting sets by Murakami and Uno~\cite{DBLP:journals/dam/MurakamiU14}. In Section~\ref{sec:experiments}, we compare the performance of our algorithm to the discovery algorithm used in~\cite{DBLP:journals/pvldb/ChuIP13,DBLP:conf/dexa/PenaA18,DBLP:journals/pvldb/PenaAN19}, and show that even though our algorithm is more general, we are able to significantly reduce the running time.


\subsection{Enumerating Minimal Hitting Sets}
We now introduce the algorithm of Murakami and Uno~\cite{DBLP:journals/dam/MurakamiU14} for enumeraing minimal hitting sets. In the next subsection, we will explain how we adapt the algorithm to the approximation problem. 

{
\begin{algseries}{t}{\label{alg:mmcs} An algorithm for enumerating minimal hitting sets.}
\begin{insidealg2}{MMCS}{$S,\crit,\uncov,\cand$}{~\cite{DBLP:journals/dam/MurakamiU14}}
\IF {$\uncov = \emptyset$} 
\STATE output $S$
\STATE \textbf{return}
\ENDIF
\STATE choose a set $F$ from $\uncov$
\STATE $C=\cand\cap F$
\STATE $\cand=\cand\setminus C$
\FORALL{$e\in C$}
\STATE $\algname{UpdateCritUncov}(e,S,\crit,\uncov)$
\IF{$\crit[u]\neq\emptyset$ for each $u\in S$}
\STATE $\algname{MMCS}(S\cup\set{e},\crit,\uncov,\cand)$
\STATE $\cand=\cand\cup\set{e}$
\ENDIF
\STATE recover the changes to $\crit$ and $\uncov$ done in $8$
\ENDFOR
\STATE recover the change to $\cand$ done in $6$
\end{insidealg2}
\begin{insidesub}{UpdateCritUncov}{$e,S,\crit,\uncov$}
\FORALL{$F\in \uncov$}
\IF {$e\in F$}
\STATE $\crit[e]= \crit[e]\cup \set{F}$
\STATE $\uncov= \uncov\setminus \set{F}$
\ENDIF
\ENDFOR
\FORALL{$u\in S$}
\FORALL {$F\in \crit[u]$}
\IF {$e\in F$}
\STATE $\crit[u] = \crit[u]\setminus \set{F}$
\ENDIF
\ENDFOR
\ENDFOR
\end{insidesub}
\end{algseries}
}

The algorithm is depicted in Figure~\ref{alg:mmcs}. The input consists of a set $K$ of elements and a set $M$ of subsets of $K$. Those are used to initialize three data structures, namely $\uncov$, $\cand$ and $\crit$, maintained by the algorithm. The algorithm is a recursive algorithm that builds the hitting sets incrementally. It starts with an empty set $S$, and adds elements to $S$ until it has a nonempty intersection with each one of the subsets in $M$; that is, until $S$ is a hitting set.
The data structure $\uncov$ stores the subsets in $M$ that are not yet covered, that is, have an empty intersection with the intermediate $S$. Since we start with an empty $S$, initially, $\uncov$ contains all the subsets in $M$. The second data structure, $\cand$, stores the elements of $K$ that can be added to $S$ in the next iterations of the algorithm. Initially, $\cand$ contains every element of $K$. Finally, $\crit$ stores, for each element $e$ in the intermediate $S$, all the subsets in $M$ for which $e$ is critical (i.e., all the subsets that contain $e$, but do not contain any other element of $S$). The importance of each one of these data structures will become clear soon.

At each iteration, the algorithm selects a subset $F$ from $\uncov$. 
The goal is then to add at least one element of $F$ to $S$, so that the two sets have a nonempty intersection. In line~5 of the algorithm, we store the intersection of $F$ and $\cand$ in $C$. The set $C$ thus contains all the elements of $F$ that we are allowed to add to $S$. Then, every element of $F$ is removed from $\cand$. Some of these elements will be added back to $\cand$ later on, while some are permanently removed from this list. The idea is the following. Let $\set{e_1,\dots,e_n}$ be the set of elements in $C$. First, we add $e_1$ to $S$, and the other elements of $C$ still do not belong to $\cand$; hence, we are able to generate minimal hitting sets that contain $e_1$, but do not contain any other element of $C$. Then, we add $e_2$ to $S$ and we add $e_1$ to $\cand$ (if some condition holds, as we will explain later). Thus, we are now able to generate minimal hitting sets that contain only $e_2$, or contain both $e_2$ and $e_1$, but do not contain any other element of $C$. Then, we add $e_3$ to $S$ and both $e_1$ and $e_2$ appear in the list of candidates, and so on. This allows us to avoid generating the same hitting set twice, but it also allows us to prune branches in the search tree early on, as we now explain.

Observe that a set $S$ is a \e{minimal} hitting set only if \e{every} element of $S$ is critical to at least one subset. Thus, after adding an element $e$ of $F$ to $S$, the $\algname{UpdateCritUncov}$ subroutine is called. This subroutine updates the data structures in the following way: \e{(a)} every subset in $\uncov$ that contains $e$ is removed from $\uncov$, as it no longer has an empty intersection with $S$, \e{(b)} every subset that has been removed from $\uncov$ is added to the list of subsets for which $e$ is critical, as it does not contain any other element of $S$, and \e{(c)} for every element $u$ in $S$, and for every subset $F$ that belongs to the list of subsets for which $u$ is critical, $F$ is removed from this list if it contains $e$ (as it now contains other elements of $S$). 

The purpose of calling $\algname{UpdateCritUncov}$ is twofold. First, it updates the data structures after adding a new element to $S$. Second, it is used to prune branches in the search tree early on. In line~9 of the algorithm, after the call to the subroutine, the algorithm tests whether for every element of $S$, the list of subsets for which it is critical is nonempty. Otherwise, as explained above, this branch will never result in a minimal hitting set. Hence, if the test of line~9 fails, we recover all changes to $\crit$ and $\uncov$, and move on to the next element of $C$ in the iteration in line~7. Observe that in this case, the element $e$ is not added back to $\cand$ due to the observation that if an element is not critical for any subset w.r.t.~$S$, then it cannot be critical for any subset w.r.t.~a set $S'$ such that $S\subseteq S'$. If, on the other hand, the test of line~9 succeeds, then we add $e$ back to $\cand$; thus, it could be added to $S$ later on. Murakami and Uno~\cite{DBLP:journals/dam/MurakamiU14} proved the following about the algorithm $\algname{MMCS}$: \e{(a)} it returns only minimal hitting sets, \e{(b)} it returns all the minimal hitting sets, and \e{(c)} it returns each minimal hitting set once. \ester{Moreover, they have shown that the time complexity of the algorithm is $O({\left\lVert M \right\rVert})$ per iteration, where $\left\lVert M \right\rVert$ is the sum of sizes of sets in $M$. The same holds for the space complexity.}

\subsection{Enumerating Approximate Hitting Sets}

{
\begin{algseries}{t}{\label{alg:amhs} Enumerating minimal ADCs - main.}
\begin{insidealg}{ADCEnum}{$S,\crit,\uncov,\cand,\canCover,f,\epsilon$}
\IF {$1-f(D,S)\le\epsilon$ and $\algname{IsMinimal}(S,f,\epsilon)$}
\STATE output DC from $S$
\STATE \textbf{return}
\ENDIF
\STATE choose a set $F\in\uncov$ s.t.~$\canCover[F]=\mbox{true}$
\IF {such a set $F$ does not exist}
\STATE \textbf{return}
\ENDIF
\STATE $\cand = \cand\setminus F$
\STATE $\algname{UpdateCanCover}(\uncov,\cand,\canCover)$
\IF{$\algname{WillCover}(S,\cand,f,\epsilon)$}
\STATE $\algname{ADCEnum}(S,\crit,\uncov,\cand,\canCover)$
\ENDIF
\STATE  recover the change to $\cand$ done in $7$
\STATE  recover the change to $\canCover$ done in $8$
\STATE $C=\cand\cap F$
\STATE $\cand=\cand\setminus C$
\FORALL{$e\in C$}
\STATE $\algname{UpdateCritUncov}(e,S,\crit,\uncov)$
\IF{$\crit[u]\neq\emptyset$ for each $u\in S$}
\STATE $\algname{RemoveRedundantPreds(e,\cand)}$
\STATE $\algname{ADCEnum}(S\cup\set{e},\crit,\uncov,\cand,\canCover)$
\STATE $\cand=\cand\cup\set{e}$
\ENDIF
\STATE recover the changes to $\crit$ and $\uncov$ done in $16$
\ENDFOR
\STATE recover the change to $\cand$ done in $14$
\end{insidealg}
\vspace{-1em}
\end{algseries}
}

{
\begin{algseries}{t}{\label{alg:amhs-sub} Enumerating minimal ADCs - subroutines.}
\begin{insidesub}{IsMinimal}{$S,f,\epsilon$}
\FORALL{$e\in S$}
\IF{$1-f(D,S\setminus\set{e})\le\epsilon$}
\STATE \textbf{return false}
\ENDIF
\ENDFOR
\STATE \textbf{return true}
\end{insidesub}
\begin{insidesub}{UpdateCanCover}{$\uncov,\cand,\canCover$}
\FORALL{$F\in \uncov$}
\FORALL{$e\in \cand$}
\IF {$e\in F$}
\STATE continue outer loop
\ENDIF
\ENDFOR
\STATE $\canCover[F]=\mbox{false}$
\ENDFOR
\end{insidesub}
\begin{insidesub}{WillCover}{$S,\cand,f,\epsilon$}
\STATE $S'=S\cup \cand$
\IF {$1-f(D,S')\le\epsilon$}
\STATE \textbf{return true}
\ENDIF
\STATE \textbf{return false}
\end{insidesub}
\vspace{-1em}
\end{algseries}
}

One may suggest to adapt the algorithm of Figure~\ref{alg:mmcs} to generate minimal approximate hitting sets by modifying the base case. Instead of stopping when all the subsets have a nonempty intersection with $S$, we will stop when our condition for minimal approximate hitting sets holds (i.e., when $1-f(D,S)\le\epsilon$ for the function $f$ and threshold $\epsilon$). It is straightforward that this will return only minimal approximate hitting sets w.r.t.~$f$ and $\epsilon$, but will it return all of them? The answer to this question is negative. The problem with this approach, which also applies to many other algorithms for enumerating minimal hitting sets~\cite{DBLP:journals/siamdm/Gainer-DewarV17}, is that when we select a new subset at each iteration and try to ``hit'' it, we define a certain order over the subsets. An easy observation is that we will never return a set that has an empty intersection with the first chosen subset, even if it has a nonempty intersection with any other subset.

Our algorithm $\algname{ADCEnum}$ for enumerating minimal ADCs is depicted in Figure~\ref{alg:amhs}. We modify the algorithm $\algname{MMCS}$ in the following way. First, we change the base case, as aforementioned; that is, we print $S$ only if $1-f(D,S)\le\epsilon$. However, we also have to explicitly check for minimality before printing $S$. This is due to the fact that while a set $S$ of elements where each $e\in S$ is critical for at least one subset of $M$ is guaranteed to be minimal when considering hitting sets, this is not the case when considering approximate hitting sets, as our $S$ is allowed to have an empty intersection with some subsets of $M$. \ester{Due to the indifference to redundancy property, this condition is still necessary when considering approximate hitting sets, since we can remove elements that are not critical for any subset without affecting the set of tuple pairs that have a non-empty intersection with $S$, and, consequently, without affecting the value of the approximation function.}
However, this condition is no longer sufficient. Therefore, we check whether $S$ is minimal in the $\algname{IsMinimal}$ subroutine, depicted in Figure~\ref{alg:amhs-sub}. There, we go over all sets $S'$ of elements obtained from $S$ by removing a single element, and for each $S'$ we check whether $1-f(D,S')\le\epsilon$. Recall that the approximation functions that we consider are monotonic; hence, if for a subset $S'$ of $S$ it holds that $1-f(D,S')>\epsilon$, then we have that $1-f(D,S'')>\epsilon$ for any $S''\subset S'$, and we do not need to go over the subsets of $S$ obtained by removing more than one element.

Next, we choose a subset $F\in\uncov$ and make two recursive calls -- one that ``hits'' the chosen $F$ (i.e., adds an element of $F$ to $S$) and one that does not. We start with the second one. Observe that our algorithm contains an additional data structure, namely $\canCover$. It is used for the additional recursive call and it contains a single value, true or false, for every subset of $M$. Initially, the value is true for all subsets. The idea is the following. Whenever we choose not to hit $F$, this set remains in $\uncov$. To avoid choosing it again in a future iteration of the algorithm (which may result in an infinite recursion), we update $\canCover[F]=\mbox{false}$ in the $\algname{UpdateCanBeCovered}$ subroutine depicted in Figure~\ref{alg:amhs-sub}. However, $F$ may not be the only subset in $\uncov$ that has an empty intersection with $\cand$ after removing all the elements of $F$ from $\cand$ in line~7. Hence, in this subroutine, we mark every subset of $M$ that is still in $\uncov$ and does not contain any element of $\cand$. This way, we avoid selecting these subsets in future iterations, which significantly reduces the number of unnecessary recursive calls. We make the recursive call after checking whether it can result in an approximate hitting set. We check that in the $\algname{WillCover}$ subroutine that adds all the elements of $\cand$ to $S$ and checks whether the result $S'$ satisfies $1-f(D,S')\le\epsilon$. If this is not the case, the monotonicity property ensures that this branch will never result in an approximate hitting set \ester{(since we cannot increase the value of the approximation function by adding less predicates)}, and we do not make the recursive call. 

The second recursive call (where we hit the selected $F$) is identical to the recursive call of the original algorithm and we do not explain it again here. \ester{Note that if we did not assume indifference to redundancy, we could not prune branches based on the $\crit$ data structure (line~17) as done in the original algorithm, since it could be the case that adding predicates that are not critical for any subset actually increases the value of the approximation function (while having no impact on the set of tuple pairs satisfying the DC)}.

While the algorithm of Figure~\ref{alg:amhs} can be used as a general algorithm for enumerating minimal approximate hitting sets, there are two aspects  that are specific to our setting. First, we do not return the hitting set $S$ itself, but the DC obtained from $S$; Second, before making the recursive call of line~19, and after adding an element $u$ to $S$, 
we remove from $\cand$ all the predicates that differ from $u$ only by the operator. This way, we avoid developing branches that will result in trivial DCs, such as $\forall t,t' \neg(t[A]<t'[A]\wedge t[A]\ge t'[A])$, and avoid developing some branches that will fail the minimality condition, such as $\forall t,t' \neg(t[A]<t'[A]\wedge t[A]\le t'[A])$ (\ester{this is again based on the assumption that the approximation function is indifferent to redundancy, and the addition of the predicate $t[A]\le t'[A]$ cannot affect the value of the approximation function on a set that already contains the predicate $t[A]< t'[A]$}).

Finally, to improve the running time of the algorithm, we do not select a random set $F$ in line~4. Murakami and Uno~\cite{DBLP:journals/dam/MurakamiU14} suggested to select the set that has the minimum size intersection with the candidate list. Doing so, we minimize the number of iterations in the loop of line~15, and decrease the number of recursive calls. The problem with this approach is that when we select such a set, we remove less predicates from the candidate list in line~7; thus, the chances of the condition of line~9 to be satisfied increase. Hence, while we decrease the number of recursive calls in line~20, we increase the number of recursive calls in line~10. In our implementation, we select the set that maximizes the intersection with the candidates list. Our experiments show that this choice decreases the running times, as the total number of recursive call decreases compared to the approach in~\cite{DBLP:journals/dam/MurakamiU14}.

\subsection{Proof of Correctness}

\ester{The correctness of $\algname{ADCEnum}$ is stated in the following theorem.}

\begin{theorem}
Let $D$ be a database. Let $f$ be a valid approximation function and let $\epsilon\ge 0$. Then, the following hold for $\algname{ADCEnum}$ w.r.t.~$f$ and $\epsilon$: (a) it returns only minimal ADCs of $D$, (b) it returns all the minimal ADCs of $D$, and (c) it returns every minimal ADC of $D$ once.
\end{theorem}
\begin{proof}
\ester{The first claim is rather straightforward, as we return a set $S$ only if $1-f(D,S)\le \epsilon$ and $S$ is minimal. For the last claim, observe that the two recursive calls in each iteration cannot result in the same $S$ (since in the first one $S$ will always have an empty intersection with the selected $F$, while in the second one $S$ will have a nonempty intersection with $F$). Moreover, the first recursive call does not modify $S$, while the second one is identical to the recursive call of the algorithm $\algname{MMSC}$. We conclude that since $\algname{MMSC}$ returns every minimal hitting set once, our algorithm returns every minimal ADC once.}
We prove by induction on $n$, the depth of the recursion, that $\algname{ADCEnum}(S,\crit,\uncov,\cand,\canCover)$ returns every minimal ADC $\varphi$ that satisfies:
\begin{citemize}
\item $S\subseteq S_\varphi$ and $S_\varphi\subseteq (S\cup \cand)$,
\item $S_\varphi$ has an empty intersection with all the sets $F\in \evi(D)$ for which $\canCover[F]=\mbox{false}$.
\end{citemize}
Note that the sets for which $\canCover[F]=\mbox{true}$ can either have an empty or a nonempty intersection with $S_\varphi$. Since at the beginning, $\cand$ contains all the predicates of $\pspace_R$ and we have that $\canCover[F]= \mbox{true}$ for each $F\in \evi(D)$, we will conclude that $\algname{ADCEnum}(\emptyset,\crit,\uncov,\cand,\canCover)$ returns every ADC $\varphi$ such that $\emptyset\subseteq S_\varphi$ and $S_\varphi\subseteq \cand$; that is, all the ADCs.

For the basis of the induction, $n=0$, one possible case is that the condition of line~1 holds. Then, the constraint corresponding to $S$ itself is a minimal ADC; thus, the only $S_\varphi$ that contains $S$ such that $\varphi$ is a minimal ADC is $S$ itself, and we indeed return $S$. Note that the sets $F$ for which it holds that $\canCover[F]=\mbox{false}$ have an empty intersection with $S$, as we update the value $\canCover[F]$ for a set $F$ to false only when $\cand$ no longer contains any predicate of $F$. If the condition of line~1 does not hold, then the only other option is that the condition of line~5 holds. In this case, no $S_\varphi$ that contains $S$ is such that $\varphi$ is an ADC, as it does not holds that $1-f(D,S)\le\epsilon$ and the remaining candidate predicates do not appear in any of the remaining sets in $\uncov$.

For the inductive step, we prove that if the claim holds for all $n\in\set{1,\dots,k-1}$, it also holds for $n = k$. Let us consider an iteration from which the depth of the recursion is $k$. In line~4, we choose a set $F$ for which $\canCover[F]=\mbox{true}$. Each $S_\varphi$ that contains $S$ either has a nonempty intersection with $F$ or an empty one. In the first case, let $S_\varphi$ be a minimal ADC that has a nonempty intersection with $F$ and satisfies the two conditions. In line~14, we go over all the predicates of $F$ and try to add each one of them to $S$.  Clearly, $S_\varphi$ contains at least one of these predicates. Let $\set{p_1,\dots,p_k}=S_\varphi\cap F$ and assume that $p_1,\dots,p_k$ is the order by which they are selected in line~14. We claim that $S_\varphi$ is generated in the recursive call made when $p_k$ is selected in line~14.

Each predicate of $\set{p_1,\dots,p_k}$ is in $\cand$ at this point, since we assume that $S_\varphi\subseteq (S\cup\cand)$ and none of these predicates can violate the minimality condition of line~17, as this will imply that $S_\varphi$ contains an element that is not critical for any subset, which is a contradiction to the fact that $S_\varphi$ is minimal (due to indifference to redundancy). Hence, every predicate in $\set{p_1,\dots,p_k}$ is added back to $\cand$ in line~20, and since $p_k$ is the last predicate selected in the loop of line~15, in that iteration all the other predicates already belong to $\cand$. From the inductive assumption, we know that $\algname{ADCEnum}$ generates every minimal ADC that contains $S\cup \set{p_k}$, is contained in $S\cup \set{p_k}\cup \cand$, and has an empty intersection with every $F'$ for which $\canCover[F']=\mbox{false}$, when given the set $S\cup \set{p_k}$ as input; among them is $S_\varphi$ (observe that we do not change the data structure $\canCover$ for the recursive call of line~19).

For the second case, let $S_\varphi$ be a minimal ADC that has an empty intersection with $F$ and satisfies the two conditions. We claim that $S_\varphi$ is generated in the recursive call of line~10. Here, we make a recursive call with the same $S$ after removing all the predicates of $F$ from $\cand$, and updating $\canCover[F']$ to false for each $F'$ that no longer contains any predicates from $\cand$. From the inductive assumption, we know that this recursive call generates every minimal ADC that contains $S$ and has an empty intersection with $F$ (as no predicate of $F$ appears in $\cand$); among them is $S_\varphi$. That concludes our proof of correctness for the algorithm.
\end{proof}

\ester{Finally, we discuss the complexity of $\algname{ADCEnum}$. There are two components of the algorithm that affect the time complexity compared to the complexity of $\algname{MMCS}$---the additional recursive call in line~10, and the computation of the function $f$ that affects the complexity per iteration. Recall that the complexity of $\algname{MMCS}$ per iteration is $O(\left\lVert M \right\rVert)$. In our case, we have that $\left\lVert M \right\rVert$ is bounded by $|\pspace|\cdot n$, where $n$ is the number of distinct sets in $\evi(D)$. We compute the function $f$ in the algorithm $|S|+2$ times, and since $|S|$ is bounded by $|\pspace|$, we conclude that the time complexity per iteration is $O\left(|\pspace|\cdot n+|\pspace|\cdot f(|\pspace|,|\evi(D)|)\right)$, where $f(|\pspace|,|\evi(D)|)$ is the time required to compute $f$. The space complexity is not affected compared to $\algname{MMCS}$ and remains $O(|\pspace|\cdot n)$.
}

%% file: sampling.tex
\newcommand{\Prob}{\mathrm{Pr}}

\section{Mining ADCs From a Sample}\label{sec:sample}
 The input to our algorithm is the evidence set
and the complexity of building it is quadratic in the size of
the database (as we have to go over all pairs of tuples), which can be
prohibitively expensive for large databases. In this section, we show
how to use a sample from the database to produce ADCs with 
probabilistic guarantees, while avoiding the cost of building the
evidence set for the entire database~\cite{heidari2019approximate}.
For simplicity, we limit our discussion to a simple approximation
function, namely, the function $f_1$ introduced in 
Section~\ref{sec:functions}. Recall that the function $f_1$ is based
on the number of tuple pairs violating the DC in the database. 

Let $J$ be a sample uniformly drawn from a database $D$ and let $\epsilon\ge 0$. Let $\phi$ be a DC. We address the following problems: \e{(1)} how to estimate the number of violations of $\phi$ in $D$ from $J$; and \e{(2)} how to use this estimate to decide on the right threshold (or approximation function) to use when enumerating ADCs from $J$.


\subsection{Estimating the Number of Violations}
\label{sec:vionum}
Since we consider the function $f_1$ that is based on
the number of violations of the DC in the database, we now show how to
estimate this number from a sample $J$ uniformly drawn from $D$. 
We represent the violations of an ADC $\phi$ as a  conflict graph
$G(V,E)$~\cite{DBLP:journals/iandc/ChomickiM05}, where $V$ is the set of vertices corresponding to the tuples
in $D$, and $E$ is the set of edges corresponding to violations of the
DC, where an edge $(t_1,t_2)$ exists if the pair $\langle t_1,t_2\rangle$ violates the DC. Note that this is a directed graph since a pair $\langle t_1,t_2\rangle$ may violate a DC that is satisfied by $\langle t_2,t_1\rangle$. \ester{Hence, the problem that we consider here is that of estimating the \e{density} of a graph from a given sample}.

\ester{To the best of our knowledge, most works on the density of random graphs focus on the generation of samples with density requirements~\cite{boulle2015universal,lorrain1971structural,arabie1978constructing,schaeffer2007graph,fortunato2010community,holland1983stochastic}, which seems to be a harder problem. Hence, the methods proposed in these works are too robust for our problem, which reflects in the high computational complexity of the proposed solutions. In our case, the graph that we obtain is different for every DC, and we need to estimate the density for a different graph in every iteration of the algorithm; hence, using solutions with a high computational cost is infeasible.
There is also a line of work that focuses on the related problem of estimating the average degree of a graph, given the degree of some of the vertices~\cite{feige2006sums,goldreich2004estimating}; however, a basic requirement in the proposed solutions is to be able to query the actual degree of at least $O(\sqrt{|V|})$ vertices. To obtain this information, we will need to go over $O(|V|\cdot\sqrt{|V|})$ pairs of tuples in each iteration of the algorithm, which is again too expensive. Hence, we propose a simple method for estimating the graph density from a sample, that has no significant impact on the computational cost of our algorithm.}

Let $p=\frac{|E|}{2\cdot {|V|\choose 2}}$ (that is, $p=1-f_1(D,S_\phi)$). Let $G_J(V_J,E_J)$ be the conflict graph of $J$.
To estimate $p$ from $J$, we use the value $\hat{p}=\frac{|E_J|}{2\cdot {|V_J|\choose 2}}$.  We define the random variable $x_i$ for each pair of nodes in $V_J$, where $x_i = 1$ with probability $p$ and $x_i=0$ with probability $1-p$.
It can be easily shown that $E(\hat{p}) = p$, so it is an unbiased estimator of $p$. 
Note that we do not make assumptions about the structure of the conflict graph or about the dependencies between the edges. 


 We further derive error bounds on our estimator to help us derive our guarantees. To compute error bounds, various methods can be used, including Chebyshev's inequality and the normal distribution assumption. Most of them require estimating the variance of our estimator. Here, we use Chebyshev's inequality:
\begin{equation*}
\Prob(|\hat{p}-E(\hat{p})|>a)\leq \frac{1}{a^2}\cdot\var(\hat{p})
\end{equation*}

We know that $p=E(\hat{p})$; hence, we now compute an upper bound on $\var(\hat{p})$. 

\begin{align*}
\var(\hat{p})&=\var\left(\frac{|E_J|}{{|V_J|\choose 2}}\right)=\frac{1}{{|V_J|\choose 2}^2}\left[E(E_J^2)-E(E_J)^2\right]\\
&=\frac{1}{{|V_J|\choose 2}^2}\left[E(E_J^2)-{|V_J|\choose 2}^2\cdot p^2\right]
\end{align*} 

We now expand the term $E(E_J^2)$ using the random variables $x_1,\dots,x_{|V_J|\choose 2}$ as follows. 
\begin{align*}
E(E_J^2)&=E\left(\left(\sum\limits_{i=1}^{|V_J|\choose 2}x_i\right)^2\right)=\\
&=\sum\limits_{i=1}^{|V_J|\choose 2}E(x_i^2)+\sum\limits_{i\neq j\in \{1,\dots,{|V_J|\choose 2}\}} E(x_i\cdot x_j)
\end{align*}

Since we do not assume anything about the dependencies between the variables $x_i$, we cannot calculate the exact value of $E(x_i\cdot x_j)$; however, we can derive an upper bound for this value. We know that $x_i\cdot x_j=1$ if and only if $x_i=x_j=1$, and the value of $E(x_i\cdot x_j)$ depends of the number of these events. Hence, if we can find an upper bound for the probability of $x_i\cdot x_j=1$, this will be an upper bound for $E(x_i\cdot x_j)$. We have the following.
\begin{align*}
& E(x_i\cdot x_j)=\Prob(x_i=1,x_j=1)=\\
&=\Prob(x_i=1|x_j=1)\cdot \Prob(x_j=1)\leq
 \Prob(x_j=1)=p
\end{align*}

Clearly, for $i=j$ it holds that $\Prob(x_i=1,x_j=1)=p$. Hence, we obtain the following upper bound on $E(E_J^2)$.
\begin{align*}
& E(E_J^2)=\sum\limits_{i=1}^{|V_J|\choose 2}E(x_i^2)+\sum\limits_{i\neq j\in \{1,\dots,{|V_J|\choose 2}\}} E(x_i\cdot x_j)=\\
&\leq {|V_J|\choose 2}\cdot p + {{|V_J|\choose 2}\choose 2}\cdot p
\end{align*}

Next, we use the upper bound on $E(E_J^2)$ to obtain an upper bound for $\var(\hat{p})$.
\begin{align*}
\var(\hat{p})\leq p\cdot\left[\frac{{|V_J|\choose 2}+{{|V_J|\choose 2}\choose 2}}{{|V_J|\choose 2}^2}-p\right]
\end{align*} 

Using Chebyshev's inequality we obtain the following:

\begin{align*}
\Prob(|\hat{p}-p|>a)\leq \frac{p}{a^2}\cdot\left[\frac{{|S|\choose 2}+{{|S|\choose 2}\choose 2}}{{|S|\choose 2}^2}-p\right]
\end{align*}
 The obtained bounds are loose since we did not assume anything about the structure of the conflict graph and the dependencies among the violations. 
We show that better bounds can be obtained under the assumption that violations (or, equivalently, edges) are introduced randomly and independently.

We first introduce the rationale behind random violations as follows. Assume a random polluter which is a probability distribution over graphs on $n$ labeled vertices, where each directed edge appears independently with probability $p$. Each violation (edge) independently occurs between two tuples without following any specific pattern. Under this assumption, the number of edges in a sample $J$ produces a binomial distribution. 
\begin{align*}
\Prob[E_J=i]={2\cdot {|V_J| \choose 2} \choose i}\cdot p^i\cdot (1-p)^{2\cdot {|V_J|\choose 2}-i}
\end{align*}


For simplicity, we assume that the sample size is not too small and $p$ is not too close to $0$ or $1$; 
hence, we can approximate the binomial $B(n,p)$ under the mentioned conditions using the normal distribution $N(np,np(1-p))$, and we can define a confidence interval parameterized by a confidence level $1-2\alpha$, and $n=2\cdot{|V_J|\choose 2}$. The confidence interval of normal distribution is given by the following equation.

\begin{equation}
\label{eq:dist}
\Prob\bigg[|p-\hat{p}|\leq z_{1-2\alpha}\cdot\sqrt{\frac{\hat{p}(1-\hat{p})}{n}}\bigg]\geq 1-2\alpha
\end{equation}

In the next subsection, we elaborate on how to use this idea to decide which threshold $\epsilon_J$ should be used on the sample, assuming that the desired threshold for the database is $\epsilon$.

\usetikzlibrary{patterns}

\subsection{Computing the Sample Threshold}
\label{sec:find}

We now focus on the following problem. Given a sample $J$, a threshold $\epsilon$ and an error bound $\alpha$, find the thresholds that should be used on the sample to obtain accurate ADCs with high probability. \ester{Note that the threshold may depend on the DC itself, since different DCs are violated by different tuple pairs, and, consequently, the conflict graphs of different DCs are different. That is, if $\varphi$ is an ADC on the sample $J$ w.r.t.~$\epsilon^\varphi_J$, then we require that with probability at least $1-\alpha$, it holds that $\varphi$ is an ADC on the entire database w.r.t.~$\epsilon$.} 
We use Inequality~\ref{eq:dist} for this task.

Using the symmetry of the normal distribution we obtain the following.
\begin{displaymath}
\Prob \bigg[p-\hat{p}\leq z_{1-2\alpha}\cdot\sqrt{\frac{\hat{p}(1-\hat{p})}{n}}\bigg]\geq 1-\alpha
\end{displaymath}

Next, we add $1$ and subtract $1$ from the expression $p-\hat{p}$ and multiply both sides of the inner inequality by $-1$. Clearly, none of these operations affects the outer inequality and we have that:
\begin{displaymath}
\Prob \bigg[(1-p)-(1-\hat{p})\geq -z_{1-2\alpha}\cdot\sqrt{\frac{\hat{p}(1-\hat{p})}{n}}\bigg]\geq 1-\alpha
\end{displaymath}

Finally, we move the term $(1-\hat{p})$ to the other side of the inner inequality to obtain the following result.
\begin{displaymath}
\Prob \bigg[(1-p)\geq (1-\hat{p})-z_{1-2\alpha}\cdot\sqrt{\frac{\hat{p}(1-\hat{p})}{n}}\bigg]\geq 1-\alpha
\end{displaymath}

Recall that our goal is to find an $\epsilon^\varphi_J$ such that if $1-\hat{p}\geq 1-\epsilon^\varphi_J$ then $\Prob(1-p\geq 1-\epsilon)>1-\alpha$. Thus, all we need to do now is to set:
\begin{displaymath}
(1-\hat{p})-z_{1-2\alpha}\cdot\sqrt{\frac{\hat{p}(1-\hat{p})}{n}}\geq 1-\epsilon
\end{displaymath}
Or, equivalently:
\begin{equation}
(1-\hat{p})\geq z_{1-2\alpha}\cdot\sqrt{\frac{\hat{p}(1-\hat{p})}{n}}+(1-\epsilon)
\label{eq:eps}
\end{equation}

Consequently, if we define $\epsilon^\varphi_J=1- z_{1-2\alpha}\cdot\sqrt{\frac{\hat{p}(1-\hat{p})}{n}}+(1-\epsilon)$, and accept the DC $\varphi$ if $1-\hat{p}\ge 1-\epsilon^\varphi_J$, then  with probability at least $1-\alpha$, this DC is an ADC on the entire database w.r.t.~the threshold $\epsilon$.
We conclude that we can use inequality~\ref{eq:eps} as criteria for accepting or rejecting an 
ADC on the sample.

\ester{Note that we can also look at Inequality~\ref{eq:eps} from a different point of view. Rather than defining a different threshold $\epsilon^\varphi_J$ for every DC, we can define the following approximation function:
$$f_1'=(1-\hat{p})-z_{1-2\alpha}\cdot\sqrt{\frac{\hat{p}(1-\hat{p})}{n}}$$
Then, Inequality~\ref{eq:eps} implies that the DC $\varphi$ is an ADC on the entire database w.r.t.~the threshold $\epsilon$ if it is an ADC on the sample w.r.t.~the approximation function $f_1'$ and the same $\epsilon$. Note that as the size of the sample increases, the value $n$ increases as well, and the difference between $f_1$ and $f_1'$ becomes very small, as expected.
}

%% file: experiments.tex
\section{Experimental Evaluation}\label{sec:experiments}

\definecolor{mypink}{RGB}{255,20,147}
\definecolor{myred}{RGB}{255,0,0}
\definecolor{mygreen}{RGB}{0,128,0}
\definecolor{myorange}{RGB}{255,140,0}
\definecolor{mypurple}{RGB}{128,0,128}
\definecolor{mybrown}{RGB}{139,69,19}
  
  \begin{table}[t]
  \small
    \centering
\begin{tabular}{|c|c|c|c|}
\hline
\textbf{Dataset} & \textbf{\#Tuples} & \textbf{\#Attributes} & \textbf{\#Golden DCs}\\
\hline\hline
Tax & $1$M & $15$ & $9$\\\hline
Stock & $123$K & $7$ & $6$\\\hline
Hospital & $115$K & $19$ & $7$\\\hline
Food & $200$K & $17$ & $10$\\\hline
Airport & $55$K & $12$ & $9$\\\hline
Adult & $32$K & $15$ & $3$\\\hline
Flight & $582$K & $20$ & $13$\\\hline
Voter & $950$K & $25$ & $12$\\\hline
\end{tabular}
\caption{Datasets.\label{table:datasets}}
\vspace{-1em}
\end{table}
  
In this section, we provide an experimental evaluation of our ADC discovery algorithm.

\subsection{Experimental Setup}
We implemented our enumeration algorithm, including the functions $f_1$ and $f_2$ in Java. As explained in Section~\ref{sec:functions}, the function $f_3$ is hard to compute for DCs; hence, we implemented the algorithm of Figure~\ref{alg:f3_appx}, and we refer to this algorithm when mentioning the function $f_3$. We also used the Java implementation of the algorithm $\algname{AFASTDC}$ by Chu et al.~\cite{DBLP:journals/pvldb/ChuIP13} and the Java implementation of the algorithm $\algname{DCFinder}$ provided by the authors of~\cite{DBLP:journals/pvldb/PenaAN19}.

All experiments were executed on a machine with an Intel Xeon CPU E5-2603 v3 (1.60GHz, 12 cores) with 64GB of RAM running Ubuntu 14.04.3 LTS. All the experiments were repeated ten times and the average values are reported.

\ester{Following previous works on the problem of discovering DCs~\cite{DBLP:journals/pvldb/ChuIP13,DBLP:journals/pvldb/BleifussKN17,DBLP:journals/pvldb/PenaAN19}, we evaluate our algorithm on seven real-world datasets (\textbf{SP Stock}, \textbf{Hospital}, \textbf{Food Inspection}, \textbf{Airport}, \textbf{Adult}, \textbf{Flight}, and \textbf{NCVoter}), and one synthetic dataset (\textbf{Tax}).
Table~\ref{table:datasets} depicts the number of tuples, attributes, and golden DCs (i.e., DCs obtained by human experts) for each one of the datasets.}

\subsection{Running Time}

 \begin{figure}[b!]
    \centering
    \scalebox{0.9}{
\begin{tikzpicture}
\begin{axis}[
    ybar=0pt,
    ymax=7000,
    ylabel={Runtime (seconds)},
    xtick={0,1,2,3,4,5,6,7},
    xticklabels={Tax,Stock,Hospital,Food,Airport,Adult,Flight,Voter},
    bar width=0.3,
    height=0.2\textwidth,
    width=0.53\textwidth,
    ymode=log,
    log basis y={10},
    tick label style={font=\small},
    label style={font=\small},
    enlargelimits=0.1
    ]
\addplot[fill=orchid,postaction={
         pattern=horizontal lines}] coordinates {(0,2371.52) (1,79.43) (2,503.040) (3,77.88) (4,13.64) (5,20.82) (6,855.23) (7,1795.94)};\label{bar-adc}
\addplot[fill=lightpink,postaction={pattern=crosshatch dots}] coordinates {(0,5749.11) (1,241.31) (2,1498.387) (3, 80) (4,31.51) (5,29.47) (6,3654.292) (7,6420.846)};\label{bar-smc}
\legend{}
\end{axis}
\end{tikzpicture}}
\vspace{-1em}
\caption{Running times of $\algname{ADCEnum}$ (\ref{bar-adc}) and $\algname{SearchMC}$ (\ref{bar-smc}).}
\label{fig:ADCEnumVSSMC}
\vspace{-1em}
  \end{figure}
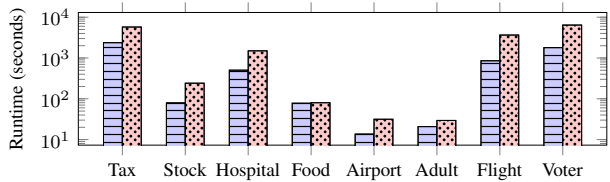
  
   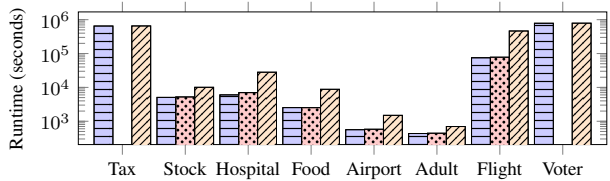
\begin{figure}[b!]
    \centering
    \scalebox{0.9}{
\begin{tikzpicture}
\begin{axis}[
    ybar=0pt,
    ymax=800000,
    ylabel={Runtime (seconds)},
    xtick={0,1,2,3,4,5,6,7},
    xticklabels={Tax,Stock,Hospital,Food,Airport,Adult,Flight,Voter},
    bar width=0.3,
    height=0.2\textwidth,
    width=0.53\textwidth,
    ymode=log,
    log basis y={10},
    tick label style={font=\small},
    label style={font=\small},
    enlargelimits=0.1
    ]
\addplot[fill=orchid,postaction={
         pattern=horizontal lines}] coordinates {(0,650623.87) (1,5070.65) (2,6020.341) (3,2529.93) (4,563.01) (5,434.48) (6,75162.57) (7,784279.888)};\label{bar-adc}
\addplot[fill=lightpink,postaction={pattern=crosshatch dots}] coordinates {(0,0) (1,5232.53) (2,7015.688) (3, 2520.96) (4,580.88) (5,443.13) (6,77961.629) (7,0)};\label{bar-smc}
\addplot[fill=lightorange,postaction={
         pattern=north east lines}] coordinates {(0,654001.46) (1,10064.74) (2,28077.167) (3,8782.4) (4,1495.26) (5,699.38) (6,463831.841) (7,788904.787)};\label{bar-afastdc}
\legend{}
\end{axis}
\end{tikzpicture}}
\vspace{-1em}
\caption{Running times of $\algname{ADCMiner}$ (\ref{bar-adc}), \algname{DCFinder} (\ref{bar-smc}), and $\algname{AFASTDC}$ (\ref{bar-afastdc}).}
\label{fig:threealg}
\vspace{-1em}
  \end{figure}

We evaluate the running time of our algorithm on the aforementioned datasets and compare them to the running times of the algorithm $\algname{AFASTDC}$~\cite{DBLP:journals/pvldb/ChuIP13}. As we do not propose a new technique for constructing the evidence set, we only compare the running times of the DC enumeration algorithms (that is, we compare our algorithm $\algname{ADCEnum}$ with the algorithm $\algname{SearchMinimalCovers}$ used in~\cite{DBLP:journals/pvldb/ChuIP13,DBLP:conf/dexa/PenaA18,DBLP:journals/pvldb/PenaAN19}, that we denote here by $\algname{SearchMC}$). We discuss the running time of the evidence set construction later.

In the experiments, we used the approximation function $f_1$ (which is the function $\algname{SearchMC}$ is designed for) with the threshold $\epsilon=0.1$.
Figure~\ref{fig:ADCEnumVSSMC} depicts the running times of both algorithms. Note that the y axis is in log scale. The results show that our algorithm is two to three times faster than $\algname{SearchMC}$ on most of the datasets. As an example, it took $\algname{SearchMC}$ 5750 seconds (96 minutes) to generate all ADCs on the entire Tax dataset, while $\algname{ADCEnum}$ finished after 2373 seconds (39 minutes); that is, about $2.5$ times faster.

We have also conducted a running time comparison between \algname{ADCEnum} and \algname{SearchMC} on different sample sizes. The results are depicted in Figure~\ref{fig:ADCvsAFASTDC2}. Note that in some cases, the running times for higher sample sizes are slightly higher than the running times for smaller sample sizes (e.g., the running time on a sample that consists of $60\%$ of the tuples in the Hospital dataset is higher than the running time on a sample that consists of $80\%$ of the tuples). This is due to the fact that while increasing the number of tuples in the database significantly increases the number of tuple pairs, which, in turn, significantly increases the total running time, the number of \e{distinct} sets in the evidence set becomes relatively stable at some point, and does not change much when more tuples are added to the database. Since the running time of ADCEnum depends on the number of distinct sets in the evidence set, the running times for different sample sizes are usually very close. Therefore, we do not reason about these small differences.

   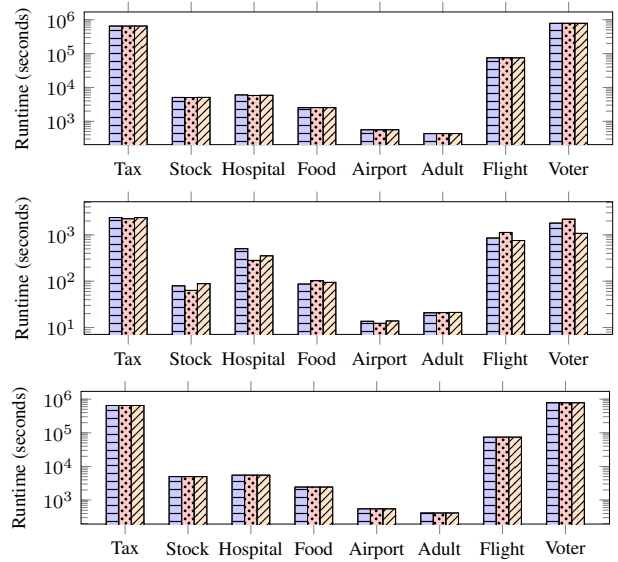
\begin{figure}[t!]
    \centering
    \scalebox{0.9}{
\begin{tikzpicture}
\begin{axis}[
    ybar=0pt,
    ymax=800000,
    ylabel={Runtime (seconds)},
    xtick={0,1,2,3,4,5,6,7},
    xticklabels={Tax,Stock,Hospital,Food,Airport,Adult,Flight,Voter},
    bar width=0.2,
    height=0.2\textwidth,
    width=0.53\textwidth,
    ymode=log,
    log basis y={10},
    tick label style={font=\small},
    label style={font=\small},
    enlargelimits=0.1
    ]
\addplot[fill=orchid,postaction={
         pattern=horizontal lines}] coordinates {(0,650623.87) (1,5070.65) (2,6020.341) (3,2529.93) (4,563.01) (5,434.48) (6,75162.57) (7,784279.888)};\label{f1_time}
\addplot[fill=lightpink,postaction={pattern=crosshatch dots}] coordinates {(0,650486.8
) (1,5054.54) (2,5798.591) (3, 2545.9
) (4,561.69) (5,434.7
) (6,75436.118) (7,784666.41)};\label{f2_time}
\addplot[fill=lightorange,postaction={
         pattern=north east lines}] coordinates {(0,650615.472) (1,5079.797) (2,5870.758) (3,2537.43) (4,563.24) (5,434.84) (6,75064.011) (7,783562.249)};\label{f3_time}
\legend{}
\end{axis}
\end{tikzpicture}}
\scalebox{0.9}{
\begin{tikzpicture}
\begin{axis}[
    ybar=0pt,
    ymax=3000,
    ylabel={Runtime (seconds)},
    xtick={0,1,2,3,4,5,6,7},
    xticklabels={Tax,Stock,Hospital,Food,Airport,Adult,Flight,Voter},
    bar width=0.2,
    height=0.2\textwidth,
    width=0.53\textwidth,
    ymode=log,
    log basis y={10},
    tick label style={font=\small},
    label style={font=\small},
    enlargelimits=0.1
    ]
\addplot[fill=orchid,postaction={
         pattern=horizontal lines}] coordinates {(0,2371.52) (1,79.43) (2,503.040) (3,86.85) (4,13.64) (5,20.82
) (6,855.233) (7,1795.947)};\label{f1_time}
\addplot[fill=lightpink,postaction={pattern=crosshatch dots}] coordinates {(0,2234.45
) (1,63.32) (2,281.29) (3, 102.82
) (4,12.32) (5,21.04
) (6,1128.781) (7,2182.469)};\label{f2_time}
\addplot[fill=lightorange,postaction={
         pattern=north east lines}] coordinates {(0,2363.122) (1,88.577) (2,353.457) (3,94.35) (4,13.87) (5,21.18) (6,756.674) (7,1078.308)};\label{f3_time}
\legend{}
\end{axis}
\end{tikzpicture}}
\scalebox{0.9}{
\begin{tikzpicture}
\begin{axis}[
    ybar=0pt,
    ymax=800000,
    ylabel={Runtime (seconds)},
    xtick={0,1,2,3,4,5,6,7},
    xticklabels={Tax,Stock,Hospital,Food,Airport,Adult,Flight,Voter},
    bar width=0.2,
    height=0.2\textwidth,
    width=0.53\textwidth,
    ymode=log,
    log basis y={10},
    tick label style={font=\small},
    label style={font=\small},
    enlargelimits=0.1
    ]
\addplot[fill=orchid,postaction={
         pattern=horizontal lines}] coordinates {(0,648252.35) (1,4991.22) (2,5517.301) (3,2443.08) (4,549.37) (5,413.66
) (6,74307.337) (7,782483.94)};
\addplot[fill=lightpink,postaction={pattern=crosshatch dots}] coordinates {(0,648252.35) (1,4991.22) (2,5517.301) (3,2443.08) (4,549.37) (5,413.66
) (6,74307.337) (7,782483.94)};
\addplot[fill=lightorange,postaction={
         pattern=north east lines}] coordinates {(0,648252.35) (1,4991.22) (2,5517.301) (3,2443.08) (4,549.37) (5,413.66
) (6,74307.337) (7,782483.94)};
\legend{}
\end{axis}
\end{tikzpicture}}
\vspace{-1em}
\caption{Running times of \algname{ADCMiner} for $f_1$ (\ref{f1_time}), $f_2$ (\ref{f2_time}), and $f_3$ (\ref{f3_time}). Top: total running time, middle: running time of ADCEnum, bottom: running time of evidence set construction.}
\label{fig:runtime}
  \end{figure}

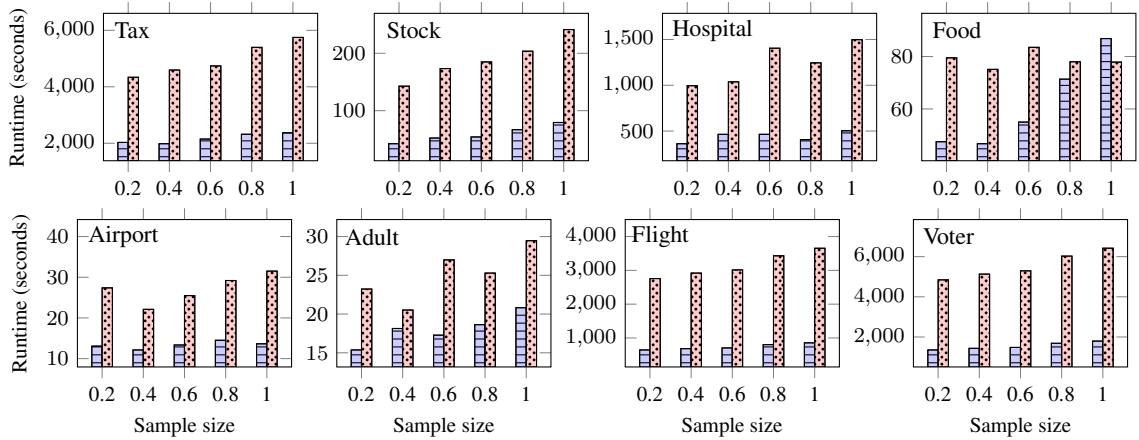
\begin{figure*}[t!]
    \centering
\begin{tikzpicture}
\begin{axis}[
    ybar=0pt,
    ymax=6000,
    enlargelimits=0.15,
    ylabel={Runtime (seconds)},
    xtick={0,1,2,3,4},
    xticklabels={0.2,0.4,0.6,0.8,1},
    bar width=0.25,
    height=0.2\textwidth,
    width=0.25\textwidth,
        tick label style={font=\small},
    label style={font=\small},
    ]
\addplot[fill=orchid,postaction={
         pattern=horizontal lines}] coordinates {(0,2029.13) (1,1983.41) (2,2148.56) (3,2319.32) (4,2371.52)};\label{bar-adc}
\addplot[fill=lightpink,postaction={pattern=crosshatch dots}] coordinates {(0,4337.47) (1,4593.73) (2,4739.63) (3,5396.65) (4,5749.11)};\label{bar-smc}
\legend{}
\end{axis}
\node[align=center, text=black] at (rel axis cs:0.25,1) {Tax};
\end{tikzpicture}
\begin{tikzpicture}
\begin{axis}[
    ybar=0pt,
    ymax=240,
    enlargelimits=0.15,
    xtick={0,1,2,3,4},
    xticklabels={0.2,0.4,0.6,0.8,1},
    bar width=0.25,
    height=0.2\textwidth,
    width=0.25\textwidth,
        tick label style={font=\small},
    label style={font=\small},
    ]
\addplot[fill=orchid,postaction={
         pattern=horizontal lines}] coordinates {(0,42.42) (1,52.31) (2,54.21) (3,66.89) (4,79.43)};
\addplot[fill=lightpink,postaction={pattern=crosshatch dots}] coordinates {(0,142.72) (1,173.38) (2,184.96) (3,203.63) (4,241.31)};
\legend{}
\end{axis}
\node[align=center, text=black] at (rel axis cs:0.3,1) {Stock};
\end{tikzpicture}
\begin{tikzpicture}
\begin{axis}[
    ybar=0pt,
    ymax=1600,
    enlargelimits=0.15,
    xtick={0,1,2,3,4},
    xticklabels={0.2,0.4,0.6,0.8,1},
    bar width=0.25,
    height=0.2\textwidth,
    width=0.25\textwidth,
        tick label style={font=\small},
    label style={font=\small},
    ]
\addplot[fill=orchid,postaction={
         pattern=horizontal lines}] coordinates {(0,361.520) (1,464.015) (2,464.229) (3,406.335) (4,503.040)};
\addplot[fill=lightpink,postaction={pattern=crosshatch dots}] coordinates {(0,995.302) (1,1037.522) (2,1404.751) (3,1245.068) (4,1498.387)};
\end{axis}
\node[align=center, text=black] at (rel axis cs:0.35,1) {Hospital};
\end{tikzpicture}
\begin{tikzpicture}
\begin{axis}[
    ybar=0pt,
    ymax=90,
    enlargelimits=0.15,
    xtick={0,1,2,3,4},
    xticklabels={0.2,0.4,0.6,0.8,1},
    bar width=0.25,
    height=0.2\textwidth,
    width=0.25\textwidth,
        tick label style={font=\small},
    label style={font=\small},
    ]
\addplot[fill=orchid,postaction={
         pattern=horizontal lines}] coordinates {(0,47.49) (1,46.73) (2,54.99) (3,71.37) (4,86.85)};
\addplot[fill=lightpink,postaction={pattern=crosshatch dots}] coordinates {(0,79.50) (1,75.10) (2,83.51) (3,78.01) (4,77.88)};
\legend{}
\end{axis}
\node[align=center, text=black] at (rel axis cs:0.28,1) {Food};
\end{tikzpicture}

\begin{tikzpicture}
\begin{axis}[
    ybar=0pt,
    ymax=40,
    enlargelimits=0.15,
    ylabel={Runtime (seconds)},
    xlabel={Sample size},
    xtick={0,1,2,3,4},
    xticklabels={0.2,0.4,0.6,0.8,1},
    bar width=0.25,
    height=0.2\textwidth,
    width=0.25\textwidth,
        tick label style={font=\small},
    label style={font=\small},
    ]
\addplot[fill=orchid,postaction={
         pattern=horizontal lines}] coordinates {(0,13.07) (1,12.13) (2,13.35) (3,14.51) (4,13.64)};
\addplot[fill=lightpink,postaction={pattern=crosshatch dots}] coordinates {(0,27.39) (1,22.11) (2,25.49) (3,29.2) (4,31.51)};
\legend{}
\end{axis}
\node[align=center, text=black] at (rel axis cs:0.33,1) {Airport};
\end{tikzpicture}
\begin{tikzpicture}
\begin{axis}[
    ybar=0pt,
    ymax=30,
    enlargelimits=0.15,
    xlabel={Sample size},
    xtick={0,1,2,3,4},
    xticklabels={0.2,0.4,0.6,0.8,1},
    bar width=0.25,
    height=0.2\textwidth,
    width=0.25\textwidth,
        tick label style={font=\small},
    label style={font=\small},
    ]
\addplot[fill=orchid,postaction={
         pattern=horizontal lines}] coordinates {(0,15.37) (1,18.14) (2,17.28) (3,18.63) (4,20.82)};
\addplot[fill=lightpink,postaction={pattern=crosshatch dots}] coordinates {(0,23.23) (1,20.52) (2,26.99) (3,25.29) (4,29.47)};
\legend{}
\end{axis}
\node[align=center, text=black] at (rel axis cs:0.28,1) {Adult};
\end{tikzpicture}
\begin{tikzpicture}
\begin{axis}[
    ybar=0pt,
    ymax=4000,
    enlargelimits=0.15,
    xlabel={Sample size},
    xtick={0,1,2,3,4},
    xticklabels={0.2,0.4,0.6,0.8,1},
    bar width=0.25,
    height=0.2\textwidth,
    width=0.25\textwidth,
        tick label style={font=\small},
    label style={font=\small},
    ]
\addplot[fill=orchid,postaction={
         pattern=horizontal lines}] coordinates {(0,645.239) (1,683.360) (2,705.064) (3,802.802) (4,855.233)};
\addplot[fill=lightpink,postaction={pattern=crosshatch dots}] coordinates {(0,2757.015) (1,2919.901) (2,3012.639) (3,3430.259) (4,3654.292)};
\legend{}
\end{axis}
\node[align=center, text=black] at (rel axis cs:0.28,1) {Flight};
\end{tikzpicture}
\begin{tikzpicture}
\begin{axis}[
    ybar=0pt,
    ymax=7000,
    enlargelimits=0.15,
    xlabel={Sample size},
    xtick={0,1,2,3,4},
    xticklabels={0.2,0.4,0.6,0.8,1},
    bar width=0.25,
    height=0.2\textwidth,
    width=0.25\textwidth,
        tick label style={font=\small},
    label style={font=\small},
    ]
\addplot[fill=orchid,postaction={
         pattern=horizontal lines}] coordinates {(0,1354.969) (1,1435.021) (2,1480.599) (3,1685.843) (4,1795.947)};
\addplot[fill=lightpink,postaction={pattern=crosshatch dots}] coordinates {(0,4844.267) (1,5130.469) (2,5293.416) (3,6027.204) (4,6420.846)};
\legend{}
\end{axis}
\node[align=center, text=black] at (rel axis cs:0.28,1) {Voter};
\end{tikzpicture}
\vspace{-1em}
\caption{Running times in seconds of $\algname{ADCEnum}$ (\ref{bar-adc}) and $\algname{SearchMC}$ (\ref{bar-smc}) for varying sample sizes.}
\label{fig:ADCvsAFASTDC2}
  \end{figure*}
  
In Figure~\ref{fig:threealg}, we compare the total running times of \algname{ADCMiner}, \algname{AFASTDC}~\cite{DBLP:journals/pvldb/ChuIP13}, and \algname{DCFinder}~\cite{DBLP:journals/pvldb/PenaAN19}. Note that we do not report on the running times of \algname{DCFinder} on the Tax and Voter datasets since we were unable to generate the evidence set with their algorithm (using the parameters recommended by the authors) even when dedicating almost the entire memory of our machine to the Java heap. While our algorithm is faster than the other two algorithms, the running time is mainly affected by the evidence set construction, which has a high computational cost in all three algorithms; hence, there is no drastic difference in the running times between our algorithm and \algname{DCFinder}. Sampling allows us to significantly reduce the running times compared to the other solutions, and we show that in the next subsection.

\ester{In Figure~\ref{fig:runtime}, we present the running times of \algname{ADCMiner} on all datasets for all three approximation functions. The top, middle, and bottom diagrams depict the total running time, the running time of \algname{ADCEnum}, and the running time of the evidence set construction, respectively. Note that the running times of \algname{ADCEnum} (which is the only part that depends on the choice of the approximation function) are very close for all three functions, and the total running time mostly depends on the evidence set construction}. To construct the evidence set, we used the algorithm introduced by Pena et al.~\cite{DBLP:journals/pvldb/PenaAN19}, which is the fastest algorithm for that task. \ester{However, since, as aforementioned, their algorithm was not able to process the Tax and NCVoter datasets, for these datasets, we used the algorithm of Chu et al.~\cite{DBLP:journals/pvldb/ChuIP13} to construct the evidence set.}
While for the Adult dataset, building the entire evidence set takes seven minutes, the evidence set construction requires almost an hour and a half on the SP Stock dataset, and more than twenty hours on the Flight dataset. This highlights the importance of incorporating sampling in our algorithm, as we are able to reduce the running times by as much as $90\%$, as we explain in the next subsection.

Finally, as discussed in Section~\ref{sec:algorithm}, we do not select a random set from $\uncov$ in each iteration of \algname{ADCEnum}, but rather the set that has the maximal intersection with the candidate list, as this choice decreases the running times, compared to the approach of Murakami and Uno~\cite{DBLP:journals/dam/MurakamiU14} who select the set that minimizes this intersection. In Figure~\ref{fig:maxvsmin}, we report the running times of \algname{ADCEnum} on $60k$ tuples from the Tax, SP Stock, and Hospital datasets, for both approaches. We see that the running times are indeed lower when we choose the set with the maximal intersection, for all three approximation functions.

 \begin{figure}[b!]
    \centering
    \scalebox{0.9}{
\begin{tikzpicture}
\begin{axis}[
    ybar=0pt,
    ymax=400,
    ylabel={Runtime (seconds)},
    xtick={0,1,2},
    xticklabels={Tax,Stock,Hospital},
    bar width=0.2,
    height=0.2\textwidth,
    width=0.3\textwidth,
    tick label style={font=\small},
    label style={font=\small},
    enlarge x limits=0.2
    ]
\addplot[fill=orchid,postaction={
         pattern=horizontal lines}] coordinates {(0,159) (1,40) (2,201)};\label{bar-adc}
\addplot[fill=lightpink,postaction={pattern=crosshatch dots}] coordinates {(0,258) (1,52) (2,258) };\label{bar-smc}
\legend{}
\end{axis}
\end{tikzpicture}}
    \scalebox{0.9}{
\begin{tikzpicture}
\begin{axis}[
    ybar=0pt,
    ymax=400,
    ylabel={Runtime (seconds)},
    xtick={0,1,2},
    xticklabels={Tax,Stock,Hospital},
    bar width=0.2,
    height=0.2\textwidth,
    width=0.3\textwidth,
    tick label style={font=\small},
    label style={font=\small},
    enlarge x limits=0.2
    ]
\addplot[fill=orchid,postaction={
         pattern=horizontal lines}] coordinates {(0,148) (1,37) (2,190)};\label{bar-adc}
\addplot[fill=lightpink,postaction={pattern=crosshatch dots}] coordinates {(0,375) (1,48) (2,292) };\label{bar-smc}
\legend{}
\end{axis}
\end{tikzpicture}}
       \scalebox{0.9}{
\begin{tikzpicture}
\begin{axis}[
    ybar=0pt,
    ymax=450,
    ylabel={Runtime (seconds)},
    xtick={0,1,2},
    xticklabels={Tax,Stock,Hospital},
    bar width=0.2,
    height=0.2\textwidth,
    width=0.3\textwidth,
    tick label style={font=\small},
    label style={font=\small},
    enlarge x limits=0.2
    ]
\addplot[fill=orchid,postaction={
         pattern=horizontal lines}] coordinates {(0,129) (1,42) (2,165)};\label{bar-adc}
\addplot[fill=lightpink,postaction={pattern=crosshatch dots}] coordinates {(0,427) (1,54) (2,333) };\label{bar-smc}
\legend{}
\end{axis}
\end{tikzpicture}}
\vspace{-1em}
\caption{Running times of $\algname{ADCEnum}$ choosing the set $F$ with the maximal (\ref{bar-adc}) and minimal (\ref{bar-smc}) intersection with $\cand$, for the functions $f_1$ (top), $f_2$ (middle), and $f_3$ (bottom).}
\label{fig:maxvsmin}
\vspace{-1em}
  \end{figure}
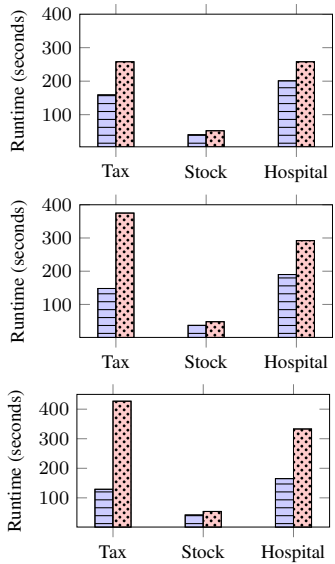

\definecolor{lightgray}{gray}{0.9}

\subsection{Sampling}

  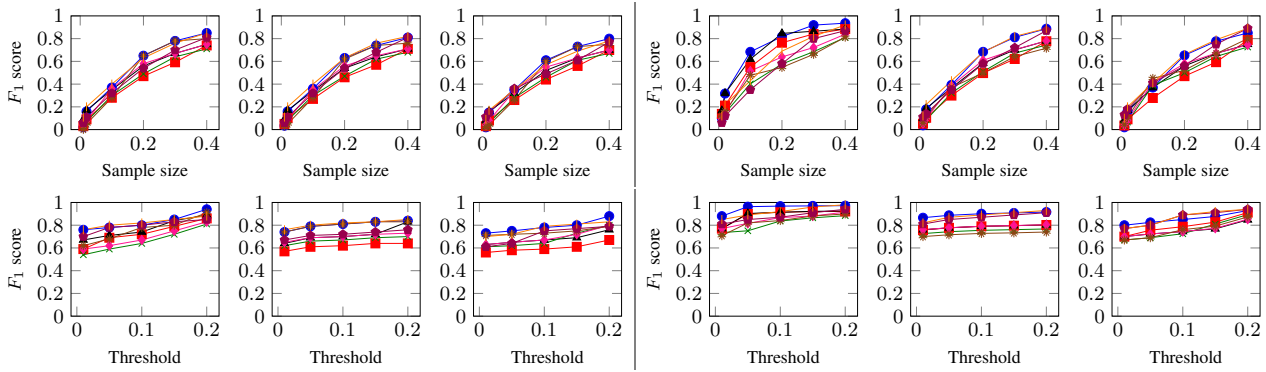
\begin{figure*}[t!]
\scalebox{0.85}{
\begin{tikzpicture}
\begin{axis}
[
xlabel = Sample size,
ytick={0,0.2,0.4,0.6,0.8,1},
ymin=0,
ymax=1,
ylabel = $F_1$ score,
width=0.22\textwidth,
label style={font=\small}
]
\addplot coordinates{
(0.01,0.02)
(0.02,0.16)
(0.1,0.37)
(0.2,0.65)
(0.3,0.78)
(0.4,0.85)
};
\addplot[color=myred,mark=square*] coordinates{
(0.01,0.03)
(0.02,0.09)
(0.1,0.28)
(0.2,0.47)
(0.3,0.592)
(0.4,0.74)
};
\addplot[color=black,mark=triangle*] coordinates{
(0.01,0.09)
(0.02,0.17)
(0.1,0.35)
(0.2,0.56)
(0.3,0.67)
(0.4,0.75)
};
\addplot[color=mygreen,mark=x] coordinates{
(0.01,0.03)
(0.02,0.07)
(0.1,0.3)
(0.2,0.49)
(0.3,0.64)
(0.4,0.71)
};
\addplot[color=mypink,mark=diamond*] coordinates{
(0.01,0.04)
(0.02,0.08)
(0.1,0.36)
(0.2,0.58)
(0.3,0.67)
(0.4,0.75)
};
\addplot[color=mybrown,mark=10-pointed star] coordinates{
(0.01,0.002)
(0.02,0.05)
(0.1,0.28)
(0.2,0.65)
(0.3,0.78)
(0.4,0.8)
};
\addplot[color=darkpurple,mark=pentagon*] coordinates{
(0.01,0.06)
(0.02,0.11)
(0.1,0.32)
(0.2,0.54)
(0.3,0.7)
(0.4,0.81)
};
\addplot[color=orange,mark=|] coordinates{
(0.01,0.11)
(0.02,0.2)
(0.1,0.4)
(0.2,0.63)
(0.3,0.76)
(0.4,0.85)
};
\end{axis}
\end{tikzpicture}
\begin{tikzpicture}
\begin{axis}
[
xlabel = Sample size,
ytick={0,0.2,0.4,0.6,0.8,1},
ymin=0,
ymax=1,
width=0.22\textwidth,
label style={font=\small}
]
\addplot coordinates{
(0.01,0.034)
(0.02,0.16)
(0.1,0.36)
(0.2,0.63)
(0.3,0.74)
(0.4,0.81)
};
\addplot[color=myred,mark=square*] coordinates{
(0.01,0.05)
(0.02,0.1)
(0.1,0.27)
(0.2,0.46)
(0.3,0.57)
(0.4,0.71)
};
\addplot[color=black,mark=triangle*] coordinates{
(0.01,0.09)
(0.02,0.17)
(0.1,0.34)
(0.2,0.54)
(0.3,0.64)
(0.4,0.71)
};
\addplot[color=mygreen,mark=x] coordinates{
(0.01,0.05)
(0.02,0.08)
(0.1,0.3)
(0.2,0.47)
(0.3,0.62)
(0.4,0.68)
};
\addplot[color=mypink,mark=diamond*] coordinates{
(0.01,0.05)
(0.02,0.09)
(0.1,0.35)
(0.2,0.56)
(0.3,0.65)
(0.4,0.71)
};
\addplot[color=mybrown,mark=10-pointed star] coordinates{
(0.01,0.03)
(0.02,0.05)
(0.1,0.28)
(0.2,0.62)
(0.3,0.74)
(0.4,0.76)
};
\addplot[color=darkpurple,mark=pentagon*] coordinates{
(0.01,0.06)
(0.02,0.11)
(0.1,0.32)
(0.2,0.54)
(0.3,0.68)
(0.4,0.81)
};
\addplot[color=orange,mark=|] coordinates{
(0.01,0.11)
(0.02,0.2)
(0.1,0.4)
(0.2,0.63)
(0.3,0.76)
(0.4,0.82)
};
\end{axis}
\end{tikzpicture}
\begin{tikzpicture}
\begin{axis}
[
xlabel = Sample size,
ytick={0,0.2,0.4,0.6,0.8,1},
ymin=0,
ymax=1,
width=0.22\textwidth,
label style={font=\small}
]
\addplot coordinates{
(0.01,0.02)
(0.02,0.15)
(0.1,0.35)
(0.2,0.61)
(0.3,0.73)
(0.4,0.8)
};
\addplot[color=myred,mark=square*] coordinates{
(0.01,0.03)
(0.02,0.08)
(0.1,0.26)
(0.2,0.44)
(0.3,0.56)
(0.4,0.7)
};
\addplot[color=black,mark=triangle*] coordinates{
(0.01,0.08)
(0.02,0.16)
(0.1,0.32)
(0.2,0.52)
(0.3,0.63)
(0.4,0.7)
};
\addplot[color=mygreen,mark=x] coordinates{
(0.01,0.03)
(0.02,0.06)
(0.1,0.28)
(0.2,0.46)
(0.3,0.61)
(0.4,0.67)
};
\addplot[color=mypink,mark=diamond*] coordinates{
(0.01,0.04)
(0.02,0.08)
(0.1,0.34)
(0.2,0.55)
(0.3,0.63)
(0.4,0.7)
};
\addplot[color=mybrown,mark=10-pointed star] coordinates{
(0.01,0.02)
(0.02,0.04)
(0.1,0.26)
(0.2,0.6)
(0.3,0.73)
(0.4,0.75)
};
\addplot[color=darkpurple,mark=pentagon*] coordinates{
(0.01,0.11)
(0.02,0.15)
(0.1,0.36)
(0.2,0.49)
(0.3,0.61)
(0.4,0.77)
};\label{f1-p7}
\addplot[color=orange,mark=|] coordinates{
(0.01,0.09)
(0.02,0.17)
(0.1,0.35)
(0.2,0.57)
(0.3,0.69)
(0.4,0.78)
};\label{f1-p8}
\end{axis}
\end{tikzpicture}
{\color{gray}\vrule}
\begin{tikzpicture}
\begin{axis}
[
xlabel = Sample size,
ytick={0,0.2,0.4,0.6,0.8,1},
ymin=0,
ymax=1,
ylabel = $F_1$ score,
width=0.22\textwidth,
label style={font=\small}
]
\addplot coordinates{
(0.01,0.1)
(0.02,0.317)
(0.1,0.684)
(0.2,0.818)
(0.3,0.919)
(0.4,0.936)
};
\addplot[color=myred,mark=square*] coordinates{
(0.01,0.138)
(0.02,0.211)
(0.1,0.55)
(0.2,0.764)
(0.3,0.841)
(0.4,0.884)
};
\addplot[color=black,mark=triangle*] coordinates{
(0.01,0.167)
(0.02,0.319)
(0.1,0.619)
(0.2,0.846)
(0.3,0.865)
(0.4,0.884)
};
\addplot[color=mygreen,mark=x] coordinates{
(0.01,0.087)
(0.02,0.145)
(0.1,0.402)
(0.2,0.577)
(0.3,0.685)
(0.4,0.81)
};
\addplot[color=mypink,mark=diamond*] coordinates{
(0.01,0.09)
(0.02,0.145)
(0.1,0.515)
(0.2,0.639)
(0.3,0.722)
(0.4,0.839)
};
\addplot[color=mybrown,mark=10-pointed star] coordinates{
(0.01,0.117)
(0.02,0.157)
(0.1,0.477)
(0.2,0.547)
(0.3,0.662)
(0.4,0.813)
};
\addplot[color=darkpurple,mark=pentagon*] coordinates{
(0.01,0.06)
(0.02,0.12)
(0.1,0.35)
(0.2,0.58)
(0.3,0.8)
(0.4,0.87)
};
\addplot[color=orange,mark=|] coordinates{
(0.01,0.12)
(0.02,0.22)
(0.1,0.43)
(0.2,0.69)
(0.3,0.82)
(0.4,0.92)
};
\end{axis}
\end{tikzpicture}
\begin{tikzpicture}
\begin{axis}
[
xlabel = Sample size,
ytick={0,0.2,0.4,0.6,0.8,1},
ymin=0,
ymax=1,
width=0.22\textwidth,
label style={font=\small}
]
\addplot coordinates{
(0.01,0.037)
(0.02,0.178)
(0.1,0.393)
(0.2,0.684)
(0.3,0.811)
(0.4,0.887)
};
\addplot[color=myred,mark=square*] coordinates{
(0.01,0.052)
(0.02,0.104)
(0.1,0.299)
(0.2,0.497)
(0.3,0.621)
(0.4,0.777)
};
\addplot[color=black,mark=triangle*] coordinates{
(0.01,0.1)
(0.02,0.186)
(0.1,0.367)
(0.2,0.584)
(0.3,0.701)
(0.4,0.78)
};
\addplot[color=mygreen,mark=x] coordinates{
(0.01,0.05)
(0.02,0.085)
(0.1,0.322)
(0.2,0.51)
(0.3,0.673)
(0.4,0.743)
};
\addplot[color=mypink,mark=diamond*] coordinates{
(0.01,0.053)
(0.02,0.097)
(0.1,0.386)
(0.2,0.61)
(0.3,0.703)
(0.4,0.78)
};
\addplot[color=mybrown,mark=10-pointed star] coordinates{
(0.01,0.06)
(0.02,0.118)
(0.1,0.37)
(0.2,0.498)
(0.3,0.642)
(0.4,0.716)
};
\addplot[color=darkpurple,mark=pentagon*] coordinates{
(0.01,0.11)
(0.02,0.14)
(0.1,0.35)
(0.2,0.58)
(0.3,0.72)
(0.4,0.87)
};
\addplot[color=orange,mark=|] coordinates{
(0.01,0.12)
(0.02,0.22)
(0.1,0.43)
(0.2,0.68)
(0.3,0.82)
(0.4,0.89)
};
\end{axis}
\end{tikzpicture}
\begin{tikzpicture}
\begin{axis}
[
xlabel = Sample size,
ytick={0,0.2,0.4,0.6,0.8,1},
ymin=0,
ymax=1,
width=0.22\textwidth,
label style={font=\small}
]
\addplot coordinates{
(0.01,0.02)
(0.02,0.161)
(0.1,0.37)
(0.2,0.653)
(0.3,0.777)
(0.4,0.85)
};
\addplot[color=myred,mark=square*] coordinates{
(0.01,0.036)
(0.02,0.09)
(0.1,0.278)
(0.2,0.47)
(0.3,0.593)
(0.4,0.788)
};
\addplot[color=black,mark=triangle*] coordinates{
(0.01,0.085)
(0.02,0.17)
(0.1,0.37)
(0.2,0.56)
(0.3,0.67)
(0.4,0.75)
};
\addplot[color=mygreen,mark=x] coordinates{
(0.01,0.035)
(0.02,0.066)
(0.1,0.39)
(0.2,0.5)
(0.3,0.643)
(0.4,0.726)
};
\addplot[color=mypink,mark=diamond*] coordinates{
(0.01,0.04)
(0.02,0.083)
(0.1,0.416)
(0.2,0.58)
(0.3,0.672)
(0.4,0.745)
};
\addplot[color=mybrown,mark=10-pointed star] coordinates{
(0.01,0.045)
(0.02,0.104)
(0.1,0.45)
(0.2,0.522)
(0.3,0.667)
(0.4,0.8)
};
\addplot[color=darkpurple,mark=pentagon*] coordinates{
(0.01,0.13)
(0.02,0.18)
(0.1,0.41)
(0.2,0.56)
(0.3,0.75)
(0.4,0.89)
};\label{f1-p7}
\addplot[color=orange,mark=|] coordinates{
(0.01,0.1)
(0.02,0.2)
(0.1,0.41)
(0.2,0.66)
(0.3,0.79)
(0.4,0.89)
};\label{f1-p8}
\end{axis}
\end{tikzpicture}}

\scalebox{0.85}{
\begin{tikzpicture}
\begin{axis}
[
xlabel = Threshold,
ytick={0,0.2,0.4,0.6,0.8,1},
ymin=0,
ymax=1,
ylabel = $F_1$ score,
width=0.22\textwidth,
label style={font=\small}
]
\addplot coordinates{
(0.01,0.76)
(0.05,0.78)
(0.1,0.8)
(0.15,0.85)
(0.2,0.94)
};\label{f1-p1}
\addplot[color=myred,mark=square*] coordinates{
(0.01,0.59)
(0.05,0.69)
(0.1,0.72)
(0.15,0.79)
(0.2,0.86)
};\label{f1-p2}
\addplot[color=black,mark=triangle*] coordinates{
(0.01,0.67)
(0.05,0.71)
(0.1,0.74)
(0.15,0.84)
(0.2,0.89)
};\label{f1-p3}
\addplot[color=mygreen,mark=x] coordinates{
(0.01,0.54)
(0.05,0.59)
(0.1,0.64)
(0.15,0.72)
(0.2,0.81)
};
\addplot[color=mypink,mark=diamond*] coordinates{
(0.01,0.59)
(0.05,0.62)
(0.1,0.67)
(0.15,0.75)
(0.2,0.83)
};
\addplot[color=mybrown,mark=10-pointed star] coordinates{
(0.01,0.62)
(0.05,0.68)
(0.1,0.78)
(0.15,0.8)
(0.2,0.91)
};
\addplot[color=darkpurple,mark=pentagon*] coordinates{
(0.01,0.7)
(0.05,0.78)
(0.1,0.8)
(0.15,0.83)
(0.2,0.85)
};\label{f1-p7}
\addplot[color=orange,mark=|] coordinates{
(0.01,0.76)
(0.05,0.8)
(0.1,0.82)
(0.15,0.85)
(0.2,0.88)
};\label{f1-p8}
\end{axis}
\end{tikzpicture}
\begin{tikzpicture}
\begin{axis}
[
xlabel = Threshold,
ytick={0,0.2,0.4,0.6,0.8,1},
ymin=0,
ymax=1,
width=0.22\textwidth,
label style={font=\small}
]
\addplot coordinates{
(0.01,0.74)
(0.05,0.79)
(0.1,0.81)
(0.15,0.83)
(0.2,0.84)
};\label{f1-p1}
\addplot[color=myred,mark=square*] coordinates{
(0.01,0.57)
(0.05,0.61)
(0.1,0.62)
(0.15,0.64)
(0.2,0.64)
};\label{f1-p2}
\addplot[color=black,mark=triangle*] coordinates{
(0.01,0.64)
(0.05,0.69)
(0.1,0.7)
(0.15,0.72)
(0.2,0.83)
};\label{f1-p3}
\addplot[color=mygreen,mark=x] coordinates{
(0.01,0.62)
(0.05,0.66)
(0.1,0.67)
(0.15,0.69)
(0.2,0.7)
};
\addplot[color=mypink,mark=diamond*] coordinates{
(0.01,0.65)
(0.05,0.69)
(0.1,0.7)
(0.15,0.72)
(0.2,0.73)
};
\addplot[color=mybrown,mark=10-pointed star] coordinates{
(0.01,0.74)
(0.05,0.79)
(0.1,0.81)
(0.15,0.83)
(0.2,0.83)
};
\addplot[color=darkpurple,mark=pentagon*] coordinates{
(0.01,0.67)
(0.05,0.71)
(0.1,0.72)
(0.15,0.74)
(0.2,0.76)
};\label{f1-p7}
\addplot[color=orange,mark=|] coordinates{
(0.01,0.76)
(0.05,0.8)
(0.1,0.82)
(0.15,0.83)
(0.2,0.85)
};\label{f1-p8}
\end{axis}
\end{tikzpicture}
\begin{tikzpicture}
\begin{axis}
[
xlabel = Threshold,
ytick={0,0.2,0.4,0.6,0.8,1},
ymin=0,
ymax=1,
width=0.22\textwidth,
label style={font=\small}
]
\addplot coordinates{
(0.01,0.73)
(0.05,0.75)
(0.1,0.78)
(0.15,0.8)
(0.2,0.88)
};\label{f1-p1}
\addplot[color=myred,mark=square*] coordinates{
(0.01,0.56)
(0.05,0.58)
(0.1,0.59)
(0.15,0.61)
(0.2,0.67)
};\label{f1-p2}
\addplot[color=black,mark=triangle*] coordinates{
(0.01,0.63)
(0.05,0.65)
(0.1,0.67)
(0.15,0.69)
(0.2,0.76)
};\label{f1-p3}
\addplot[color=mygreen,mark=x] coordinates{
(0.01,0.61)
(0.05,0.62)
(0.1,0.64)
(0.15,0.73)
(0.2,0.8)
};\label{f1-p4}
\addplot[color=mypink,mark=diamond*] coordinates{
(0.01,0.63)
(0.05,0.65)
(0.1,0.67)
(0.15,0.72)
(0.2,0.8)
};\label{f1-p5}
\addplot[color=mybrown,mark=10-pointed star] coordinates{
(0.01,0.71)
(0.05,0.72)
(0.1,0.73)
(0.15,0.75)
(0.2,0.79)
};\label{f1-p6}
\addplot[color=darkpurple,mark=pentagon*] coordinates{
(0.01,0.61)
(0.05,0.64)
(0.1,0.75)
(0.15,0.77)
(0.2,0.79)
};\label{f1-p7}
\addplot[color=orange,mark=|] coordinates{
(0.01,0.7)
(0.05,0.72)
(0.1,0.79)
(0.15,0.81)
(0.2,0.83)
};\label{f1-p8}
\end{axis}
\end{tikzpicture}
{\color{gray}\vrule}
\begin{tikzpicture}
\begin{axis}
[
xlabel = Threshold,
ytick={0,0.2,0.4,0.6,0.8,1},
ymin=0,
ymax=1,
ylabel = $F_1$ score,
width=0.22\textwidth,
label style={font=\small}
]
\addplot coordinates{
(0.01,0.88)
(0.05,0.963)
(0.1,0.968)
(0.15,0.97)
(0.2,0.972)
};\label{f1-p1}
\addplot[color=myred,mark=square*] coordinates{
(0.01,0.77)
(0.05,0.899)
(0.1,0.914)
(0.15,0.92)
(0.2,0.921)
};\label{f1-p2}
\addplot[color=black,mark=triangle*] coordinates{
(0.01,0.77)
(0.05,0.903)
(0.1,0.914)
(0.15,0.922)
(0.2,0.925)
};\label{f1-p3}
\addplot[color=mygreen,mark=x] coordinates{
(0.01,0.735)
(0.05,0.751)
(0.1,0.837)
(0.15,0.865)
(0.2,0.884)
};
\addplot[color=mypink,mark=diamond*] coordinates{
(0.01,0.771)
(0.05,0.816)
(0.1,0.867)
(0.15,0.881)
(0.2,0.915)
};
\addplot[color=mybrown,mark=10-pointed star] coordinates{
(0.01,0.708)
(0.05,0.82)
(0.1,0.841)
(0.15,0.881)
(0.2,0.899)
};
\addplot[color=darkpurple,mark=pentagon*] coordinates{
(0.01,0.81)
(0.05,0.85)
(0.1,0.87)
(0.15,0.91)
(0.2,0.94)
};\label{f1-p7}
\addplot[color=orange,mark=|] coordinates{
(0.01,0.85)
(0.05,0.89)
(0.1,0.92)
(0.15,0.96)
(0.2,0.98)
};\label{f1-p8}
\end{axis}
\end{tikzpicture}
\begin{tikzpicture}
\begin{axis}
[
xlabel = Threshold,
ytick={0,0.2,0.4,0.6,0.8,1},
ymin=0,
ymax=1,
width=0.22\textwidth,
label style={font=\small}
]
\addplot coordinates{
(0.01,0.866)
(0.05,0.887)
(0.1,0.901)
(0.15,0.908)
(0.2,0.916)
};\label{f1-p1}
\addplot[color=myred,mark=square*] coordinates{
(0.01,0.758)
(0.05,0.777)
(0.1,0.79)
(0.15,0.796)
(0.2,0.802)
};\label{f1-p2}
\addplot[color=black,mark=triangle*] coordinates{
(0.01,0.76)
(0.05,0.777)
(0.1,0.791)
(0.15,0.796)
(0.2,0.803)
};\label{f1-p3}
\addplot[color=mygreen,mark=x] coordinates{
(0.01,0.726)
(0.05,0.743)
(0.1,0.755)
(0.15,0.761)
(0.2,0.767)
};
\addplot[color=mypink,mark=diamond*] coordinates{
(0.01,0.761)
(0.05,0.778)
(0.1,0.791)
(0.15,0.797)
(0.2,0.804)
};
\addplot[color=mybrown,mark=10-pointed star] coordinates{
(0.01,0.7)
(0.05,0.716)
(0.1,0.728)
(0.15,0.733)
(0.2,0.74)
};
\addplot[color=darkpurple,mark=pentagon*] coordinates{
(0.01,0.81)
(0.05,0.85)
(0.1,0.87)
(0.15,0.89)
(0.2,0.91)
};\label{f1-p7}
\addplot[color=orange,mark=|] coordinates{
(0.01,0.82)
(0.05,0.87)
(0.1,0.89)
(0.15,0.91)
(0.2,0.93)
};\label{f1-p8}
\end{axis}
\end{tikzpicture}
\begin{tikzpicture}
\begin{axis}
[
xlabel = Threshold,
ytick={0,0.2,0.4,0.6,0.8,1},
ymin=0,
ymax=1,
width=0.22\textwidth,
label style={font=\small}
]
\addplot coordinates{
(0.01,0.8)
(0.05,0.825)
(0.1,0.85)
(0.15,0.877)
(0.2,0.929)
};\label{f1-p1}
\addplot[color=myred,mark=square*] coordinates{
(0.01,0.7)
(0.05,0.756)
(0.1,0.787)
(0.15,0.827)
(0.2,0.911)
};\label{f1-p2}
\addplot[color=black,mark=triangle*] coordinates{
(0.01,0.7)
(0.05,0.723)
(0.1,0.744)
(0.15,0.768)
(0.2,0.847)
};\label{f1-p3}
\addplot[color=mygreen,mark=x] coordinates{
(0.01,0.67)
(0.05,0.69)
(0.1,0.725)
(0.15,0.81)
(0.2,0.893)
};\label{f1-p4}
\addplot[color=mypink,mark=diamond*] coordinates{
(0.01,0.7)
(0.05,0.723)
(0.1,0.744)
(0.15,0.773)
(0.2,0.853)
};\label{f1-p5}
\addplot[color=mybrown,mark=10-pointed star] coordinates{
(0.01,0.67)
(0.05,0.69)
(0.1,0.761)
(0.15,0.792)
(0.2,0.873)
};\label{f1-p6}
\addplot[color=darkpurple,mark=pentagon*] coordinates{
(0.01,0.77)
(0.05,0.8)
(0.1,0.89)
(0.15,0.91)
(0.2,0.94)
};\label{f1-p7}
\addplot[color=orange,mark=|] coordinates{
(0.01,0.77)
(0.05,0.8)
(0.1,0.89)
(0.15,0.92)
(0.2,0.94)
};\label{f1-p8}
\end{axis}
\end{tikzpicture}
}
    \caption{$F_1$ score for varying sample sizes and fixed $\epsilon=0.01$ (top left) and $\epsilon=0.1$ (top right), and varying thresholds and fixed sample size $0.3$ (bottom left) and $0.4$ (bottom right), under $f_1$ (left), $f_2$ (middle), and $f_3$ (right). Datasets: Tax (\ref{f1-p1}), Stock (\ref{f1-p3}), Hospital (\ref{f1-p2}), Food (\ref{f1-p4}), Airport (\ref{f1-p5}), Adult (\ref{f1-p6}), Flight (\ref{f1-p7}), and Voter (\ref{f1-p8}).}
    \label{fig:sample}
\end{figure*}

We now report on the quality of the ADCs obtained from a sample. \ester{In all of our experiments, the sample size is big enough so that the term $z_{1-2\alpha}\cdot\sqrt{\frac{\hat{p}(1-\hat{p})}{n}}$ in the approximation function $f_1'$ defined in Section~\ref{sec:sample} has practically no impact on the function. Therefore, we use the same approximation function and threshold on both the sample and the entire dataset.}
In the experiments reported in Figure~\ref{fig:sample}, we use a standard measure of quality, namely the $F_1$ score (i.e., $2\cdot\frac{\mbox{precision}\times\mbox{recall}}{\mbox{precision}+\mbox{recall}}$). We compare the ADCs obtained from the sample with the ADCs obtained from the entire dataset. 

We first fix a threshold and consider different sample sizes. The first three charts on the top of Figure~\ref{fig:sample} show the $F_1$ score for a fixed threshold $\epsilon=0.01$ for varying sample sizes, ranging from $1\%$ to $40\%$ of the tuples in the dataset, for all three approximation functions. The last three charts show the $F_1$ score for a fixed threshold $\epsilon=0.1$ for varying sample sizes. Clearly, the larger the sample is, the more accurate the results we obtain. Generally, we see that in order to obtain an $F_1$ score of about $0.7$ or above we need to see about $40\%$ of the tuples in the dataset. \ester{Note that we obtain a higher $F_1$-score on larger datasets (for which sampling is particularly important), as for such datasets a relatively small sample allows us to see enough tuples to obtain accurate results. For example, on the Tax and NCVoter datasets we consistently obtain an $F_1$-score of at least $0.7$ or $0.8$ when seeing $30\%$ or $40\%$ of the tuples, respectively.}

The first three charts on the bottom of Figure~\ref{fig:sample} depict the $F_1$ score for a fixed sample size of $30\%$ and varying thresholds, ranging from $0.01$ to $0.2$, for all three approximation functions. The last three charts depict the $F_1$ score for a fixed sample size of $40\%$ and varying thresholds. Here, we can see that we obtain more accurate results when considering a higher threshold. This is due to the fact that a higher $\epsilon$ allows for more exceptions, and the DCs obtained using a higher threshold can be seen as obtained using a smaller sample of the database (as a smaller part of the database satisfies them). Hence, we are able to obtain result with high accuracy when considering a relatively small random sample. \ester{We conclude that the choice of the right threshold and sample size should be based on the size of the original dataset and the approximation function (as we discuss in the next subsection).}

\renewcommand{\arraystretch}{2}
\begin{figure}[b!]
    \centering
    \scalebox{0.9}{
\begin{tikzpicture}
\begin{axis}[
    ybar=0pt,
    ymax=800000,
    ylabel={Runtime (seconds)},
    xtick={0,1,2,3,4,5,6,7},
    xticklabels={Tax,Stock,Hospital,Food,Airport,Adult,Flight,Voter},
    bar width=0.16,
    height=0.2\textwidth,
    width=0.53\textwidth,
    ymode=log,
    log basis y={10},
    tick label style={font=\small},
    label style={font=\small},
    enlargelimits=0.1
    ]
\addplot[fill=orchid,postaction={
         pattern=horizontal lines}] coordinates {(0,16138.85) (1,835.86) (2,1299.301) (3,1348.31) (4,120.48) (5,83.79) (6,2062.16) (7,6445.49)};\label{20_time}
\addplot[fill=lightpink,postaction={pattern=crosshatch dots}] coordinates {(0,28820.62) (1,1926.83) (2,3003.279) (3, 1642.92) (4,223.66) (5,175.42) (6,4267.49) (7,13964.18)};\label{40_time}
\addplot[fill=lightorange,postaction={
         pattern=north east lines}] coordinates {(0,40430.23) (1,2679.4) (2,4068.433) (3,1763.98) (4,262.68) (5,238.72) (6,8160.796) (7,28902.189)};\label{60_time}
\addplot[fill=darkgreen] coordinates { (0,244748.11) (1,3961.02) (2,5208.571) (3,2109.52) (4,421.22) (5,346.09) (6,28389.2) (7,292180.939)};\label{80_time}
\addplot[fill=darkpurple,postaction={
         pattern=north west lines}] coordinates { (0,650623.87) (1,5070.65) (2,6020.341) (3,2529.93) (4,563.01) (5,434.48) (6,75162.57) (7,784279.888)};\label{100_time}
\legend{}
\end{axis}
\end{tikzpicture}}
\vspace{-1em}
\caption{Running times of $\algname{ADCMiner}$ for varying sample sizes---$20\%$ (\ref{20_time}), $40\%$ (\ref{40_time}), $60\%$ (\ref{60_time}), $80\%$ (\ref{80_time}), and $100\%$ (\ref{100_time}).}
\label{fig:runtime-sample}
  \end{figure}

Next, we show the improvement in running times obtained when considering a sample. Figure~\ref{fig:runtime-sample} depicts the running times of \algname{ADCMiner} for varying sample sizes on all datasets for the function $f_1$ (as shown in the previous subsection, the running times for all three functions are very close; hence, all three functions follow a similar trend). On the SP Stock dataset, we are able to reduce the total running time from eighty five to thirty two minutes when considering a sample that consists of $40\%$ of the tuples---a reduction of more than $60\%$. For the Flight dataset, the running time goes down from almost twenty one hours to seventy minutes---a reduction of almost $95\%$. For the Tax dataset, we are not able to use the same algorithm for constructing the evidence set on $100\%$ and $40\%$ of the tuples in the database. Using the original algorithm for constructing the evidence set by Chu et al.~\cite{DBLP:journals/pvldb/ChuIP13} we obtain a reduction of more than $94\%$---from $7.5$ days to $10.5$ hours. Using the algorithm $\algname{BFASTDC}$ to construct the evidence set we can obtain a similar reduction (of almost $90\%$) in the running time~\cite{DBLP:conf/dexa/PenaA18}.

\ester{Finally, we validate the theoretical analysis of Section~\ref{sec:sample} as follows. For each dataset, we run our algorithm with the approximation function $f_1$ on varying sample sizes ranging from $5\%$ to $80\%$ of the tuples in the dataset. For each such sample, we compute the average value of $\epsilon-\hat{p}$ over the discovered ADCs (recall that $\hat{p}$ is the proportion of violating tuple pairs). Figure~\ref{fig:validation} depicts the values obtained in this experiment. Note that the actual numbers are very small; hence, the reported numbers are scaled up for each dataset (i.e., multiplied by a constant $10^x$, where $x$ depends on the dataset). We see that as the sample size increases, the value $\epsilon-\hat{p}$ decreases. Moreover, for each dataset, we have that $(\epsilon-\hat{p})\sim \frac{1}{\sqrt{n}}$ (where $\sim$ denotes asymptotic equivalence and $n$ is defined as in Section~\ref{sec:sample}), which supports our main result of Section~\ref{sec:sample} (i.e., Inequality~\ref{eq:eps})}.

\subsection{Qualitative Analysis}
  


We compare the three approximation functions discussed in Section~\ref{sec:functions} in the following way. For each one of the datasets, we have a set of ``golden'' DCs; i.e., DCs obtained by domain experts. We take a sample of $10K$ tuples from each one of the datasets and add noise to the resulting dataset, such that each value has a probability of $0.001$ to be modified, and if it is modified, then it has $50\%$ chance of being changed to a new value from the active domain of the corresponding column and $50\%$ chance to being changed to a typo. We also generate another dirty dataset in a similar way, but in this case, we only allow changing values in $0.001$ of the tuples. Hence, in the first dataset, the errors are distributed among the tuples \ester{(and the number of modified tuples is usually very close to the number of modified values)}, while in the second dataset, the errors are concentrated in a small subset of the tuples.

Then, we run our algorithm on the two dirty datasets obtained from each one of the original datasets, with varying approximation thresholds $\epsilon$ (ranging from $10^{-6}$ to $10^{-1}$). For each $\epsilon$, we compute the G-recall, that is, the number of golden DC returned divided by the total number of golden DCs. We report the results in Figure~\ref{fig:app-fun}. \ester{We also report the G-recall for $\epsilon=0$ (i.e., when considering valid DCs) above each diagram (in parentheses).}
We observe the following phenomena. 
\ester{First, the G-recall for valid DCs is consistently zero, or very close to zero, which highlights the importance of considering approximate DCs.}
Second, the function $f_1$ produces results with a higher G-recall on smaller thresholds (i.e., $10^{-5}-10^{-3}$), while the other two functions have a higher G-recall on the larger thresholds (i.e., $10^{-2}-10^{-1}$). This is due to the fact that the functions $f_2$ and $f_3$ are more sensitive in the sense that a single tuple adds $\frac{1}{n}$ to the value of the functions $f_2$ and $f_3$ (where $n$ is the number of tuples), while a pair of tuples adds $\frac{1}{n^2}$ to the value of the function $f_1$.

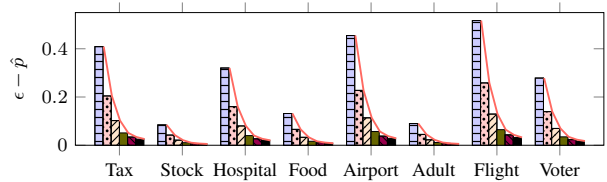
\begin{figure}[b!]
    \centering
    \scalebox{0.9}{
\begin{tikzpicture}
\begin{axis}[
    ybar=0pt,
    ymax=0.55,
    ymin=0,
    ylabel={$\epsilon-\hat{p}$},
    xtick={0,1,2,3,4,5,6,7},
    xticklabels={Tax,Stock,Hospital,Food,Airport,Adult,Flight,Voter},
    bar width=0.13,
    height=0.2\textwidth,
    width=0.53\textwidth,
    tick label style={font=\small},
    label style={font=\small},
    ]
\addplot[fill=orchid,postaction={
         pattern=horizontal lines}] coordinates {(0,0.40886) (1,0.084416) (2,0.32027) (3,0.13161) (4,0.45485) (5,0.08992) (6,0.51689) (7,0.27860)};\label{val_5}
\addplot[fill=lightpink,postaction={pattern=crosshatch dots}] coordinates {(0,0.20443) (1,0.042208) (2,0.16013) (3, 0.06581) (4,0.22742) (5,0.04496) (6,0.25844) (7,0.13930)};\label{val_10}
\addplot[fill=lightorange,postaction={
         pattern=north east lines}] coordinates {(0,0.10222) (1,0.021104) (2,0.08007) (3,0.03290) (4,0.11371) (5,0.02248) (6,0.12922) (7,0.06965)};\label{val_20}
\addplot[fill=darkgreen] coordinates { (0,0.05111) (1,0.010552) (2,0.04003) (3,0.01645) (4,0.05686) (5,0.01124) (6,0.06461) (7,0.03482)};\label{val_40}
\addplot[fill=darkpurple,postaction={
         pattern=north west lines}] coordinates { (0,0.03407) (1,0.007035) (2,0.02669) (3,0.01097) (4,0.03790) (5,0.00749) (6,0.04307) (7,0.02322)};\label{val_60}
\addplot[fill=black] coordinates { (0,0.02555) (1,0.005276) (2,0.02002) (3,0.00823) (4,0.02843) (5,0.00562) (6,0.03231) (7,0.01741)};\label{val_80}
\addplot[red,sharp plot,thick,update limits=false,] coordinates { (-0.25,0.40886) (-0.12,0.20443) (0.01,0.10222) (0.14,0.05111) (0.27,0.03407) (0.4,0.02555)};
\addplot[red,sharp plot,thick,update limits=false,] coordinates { (0.75,0.084416) (0.88,0.042208) (1.01,0.021104) (1.14,0.010552) (1.27,0.007035) (1.4,0.005276)};
\addplot[red,sharp plot,thick,update limits=false,] coordinates { (1.75,0.32027) (1.88,0.16013) (2.01,0.08007) (2.14,0.04003) (2.27,0.02669) (2.4,0.02002)};
\addplot[red,sharp plot,thick,update limits=false,] coordinates { (2.75,0.13161) (2.88,0.06581) (3.01,0.03290) (3.14,0.01645) (3.27,0.01097) (3.4,0.00823)};
\addplot[red,sharp plot,thick,update limits=false,] coordinates { (3.75,0.45485) (3.88,0.22742) (4.01,0.11371) (4.14,0.05686) (4.27,0.03790) (4.4,0.02843)};
\addplot[red,sharp plot,thick,update limits=false,] coordinates { (4.75,0.08992) (4.88,0.04496) (5.01,0.02248) (5.14,0.01124) (5.27,0.00749) (5.4,0.00562)};
\addplot[red,sharp plot,thick,update limits=false,] coordinates { (5.75,0.51689) (5.88,0.25844) (6.01,0.12922) (6.14,0.06461) (6.27,0.04307) (6.4,0.03231)};
\addplot[red,sharp plot,thick,update limits=false,] coordinates { (6.75,0.27860) (6.88,0.13930) (7.01,0.06965) (7.14,0.03482) (7.27,0.02322) (7.4,0.01741)};
\legend{}
\end{axis}
\end{tikzpicture}}
\caption{The average difference between $\epsilon$ and $\hat{p}$ over the ADCs obtained from varying sample sizes---5\% (\ref{val_5}), 10\% (\ref{val_10}), 20\% (\ref{val_20}), 40\% (\ref{val_40}), 60\% (\ref{val_60}), and , 80\% (\ref{val_80}).}
\label{fig:validation}
  \end{figure}


\begin{figure*}
  \scalebox{0.7}{
\begin{tikzpicture}
\begin{axis}
[
legend style = { at = {(0.97,0.4)}},
ylabel = G-recall,
ymax = 1,
ymin =0,
width=0.22\textwidth
]
\addplot coordinates{
(-6,0.44)
(-5,0.44)
(-4,0.67)
(-3,0.67)
(-2,0)
(-1,0)
};
\addplot coordinates{
(-6,0)
(-5,0)
(-4,0)
(-3,0.22)
(-2,0.44)
(-1,0.44)
};
\addplot[color=black,mark=triangle*] coordinates{
(-6,0)
(-5,0)
(-4,0)
(-3,0.44)
(-2,1)
(-1,0.77)
};
\end{axis}
\node[align=center, text=black] at (rel axis cs:0.19,1.13) {\textbf{Tax}};
\node[align=center, text=black] at (rel axis cs:0.4,1.13) {($0$)};
\end{tikzpicture}
\begin{tikzpicture}
\begin{axis}
[
legend style = { at = {(0.97,0.4)}},
ymax = 1,
ymin =0,
width=0.22\textwidth
]
\addplot coordinates{
(-6,0.44)
(-5,1)
(-4,0.67)
(-3,0.67)
(-2,0)
(-1,0)
};
\addplot coordinates{
(-6,0.11)
(-5,0.11)
(-4,0.11)
(-3,0.22)
(-2,0.55)
(-1,1)
};
\addplot[color=black,mark=triangle*] coordinates{
(-6,0.11)
(-5,0.11)
(-4,0.11)
(-3,1)
(-2,1)
(-1,0.78)
};
\end{axis}
\node[align=center, text=black] at (rel axis cs:0.25,1.13) {($0.11$)};
\end{tikzpicture}
{\color{gray}\vrule}
\begin{tikzpicture}
\begin{axis}
[
legend style = { at = {(0.97,0.4)}},
ymax = 1,
ymin =0,
width=0.22\textwidth
]
\addplot coordinates{
(-6,1)
(-5,1)
(-4,1)
(-3,1)
(-2,0.83)
(-1,0.83)
};
\addplot coordinates{
(-6,0)
(-5,0)
(-4,0)
(-3,1)
(-2,1)
(-1,1)
};
\addplot[color=black,mark=triangle*] coordinates{
(-6,0)
(-5,0)
(-4,0)
(-3,1)
(-2,1)
(-1,1)
};
\end{axis}
\node[align=center, text=black] at (rel axis cs:0.25,1.13) {\textbf{Stock}};
\node[align=center, text=black] at (rel axis cs:0.52,1.13) {($0$)};
\end{tikzpicture}
\begin{tikzpicture}
\begin{axis}
[
legend style = { at = {(0.97,0.4)}},
ymax = 1,
ymin =0,
width=0.22\textwidth
]
\addplot coordinates{
(-6,1)
(-5,1)
(-4,1)
(-3,1)
(-2,0.833)
(-1,0.833)
};
\addplot coordinates{
(-6,0.17)
(-5,0.17)
(-4,0.17)
(-3,1)
(-2,1)
(-1,1)
};
\addplot[color=black,mark=triangle*] coordinates{
(-6,0.17)
(-5,0.17)
(-4,0.17)
(-3,1)
(-2,1)
(-1,1)
};
\end{axis}
\node[align=center, text=black] at (rel axis cs:0.16,1.13) {($0$)};
\end{tikzpicture}
{\color{gray}\vrule}
\begin{tikzpicture}
\begin{axis}
[
legend style = { at = {(0.97,0.4)}},
ymax = 1,
ymin =0,
width=0.22\textwidth
]
\addplot coordinates{
(-6,0)
(-5,0.57)
(-4,0.71)
(-3,1)
(-2,0.28)
(-1,0)
};
\addplot coordinates{
(-6,0)
(-5,0)
(-4,0)
(-3,0)
(-2,0.14)
(-1,0.57)
};
\addplot[color=black,mark=triangle*] coordinates{
(-6,0)
(-5,0)
(-4,0)
(-3,0.14)
(-2,1)
(-1,1)
};
\end{axis}
\node[align=center, text=black] at (rel axis cs:0.32,1.13) {\textbf{Hospital}};
\node[align=center, text=black] at (rel axis cs:0.68,1.13) {($0$)};
\end{tikzpicture}
\begin{tikzpicture}
\begin{axis}
[
legend style = { at = {(0.97,0.4)}},
ymax = 1,
ymin =0,
width=0.22\textwidth
]
\addplot coordinates{
(-6,0.285)
(-5,0.714)
(-4,1)
(-3,1)
(-2,0.285)
(-1,0)
};\label{functions-p1}
\addplot coordinates{
(-6,0)
(-5,0)
(-4,0)
(-3,0)
(-2,0.57)
(-1,0.85)
};\label{functions-p2}
\addplot[color=black,mark=triangle*] coordinates{
(-6,0)
(-5,0)
(-4,0)
(-3,1)
(-2,1)
(-1,1)
};\label{functions-p3}
\end{axis}
\node[align=center, text=black] at (rel axis cs:0.15,1.13) {($0$)};
\end{tikzpicture}
{\color{gray}\vrule}
\begin{tikzpicture}
\begin{axis}
[
legend style = { at = {(0.97,0.4)}},
ymax = 1,
ymin =0,
width=0.22\textwidth
]
\addplot coordinates{
(-6,0.1)
(-5,1)
(-4,0.1)
(-3,0.1)
(-2,0.1)
(-1,0.1)
};\label{functions-p1}
\addplot coordinates{
(-6,0)
(-5,0)
(-4,0)
(-3,0.1)
(-2,0.8)
(-1,1)
};\label{functions-p2}
\addplot[color=black,mark=triangle*] coordinates{
(-6,0)
(-5,0)
(-4,0)
(-3,0.1)
(-2,1)
(-1,1)
};\label{functions-p3}
\end{axis}
\node[align=center, text=black] at (rel axis cs:0.22,1.13) {\textbf{Food}};
\node[align=center, text=black] at (rel axis cs:0.47,1.13) {($0$)};
\end{tikzpicture}
\begin{tikzpicture}
\begin{axis}
[
legend style = { at = {(0.97,0.4)}},
ymax = 1,
ymin =0,
width=0.22\textwidth
]
\addplot coordinates{
(-6,0.1)
(-5,1)
(-4,0.1)
(-3,0.1)
(-2,0.1)
(-1,0.1)
};\label{functions-p1}
\addplot coordinates{
(-6,0)
(-5,0)
(-4,0)
(-3,0.1)
(-2,0.8)
(-1,1)
};\label{functions-p2}
\addplot[color=black,mark=triangle*] coordinates{
(-6,0)
(-5,0)
(-4,0)
(-3,0.2)
(-2,1)
(-1,1)
};\label{functions-p3}
\end{axis}
\node[align=center, text=black] at (rel axis cs:0.15,1.13) {($0$)};
\end{tikzpicture}
}
\scalebox{0.7}{
\begin{tikzpicture}
\begin{axis}
[
legend style = { at = {(0.97,0.4)}},
xlabel = $\epsilon$ ($10^x$),
ylabel = G-recall,
ymax = 1,
ymin =0,
width=0.22\textwidth
]
\addplot coordinates{
(-6,0.22)
(-5,0.44)
(-4,0.88)
(-3,0.66)
(-2,0.66)
(-1,0.44)
};\label{functions-p1}
\addplot coordinates{
(-6,0)
(-5,0)
(-4,0)
(-3,0.11)
(-2,0.22)
(-1,0.77)
};\label{functions-p2}
\addplot[color=black,mark=triangle*] coordinates{
(-6,0)
(-5,0)
(-4,0)
(-3,0.11)
(-2,0.77)
(-1,0.66)
};\label{functions-p3}
\end{axis}
\node[align=center, text=black] at (rel axis cs:0.3,1.13) {\textbf{Airport}};
\node[align=center, text=black] at (rel axis cs:0.63,1.13) {($0$)};
\end{tikzpicture}
\begin{tikzpicture}
\begin{axis}
[
legend style = { at = {(0.97,0.4)}},
xlabel = $\epsilon$ ($10^x$),
ymax = 1,
ymin =0,
width=0.22\textwidth
]
\addplot coordinates{
(-6,0.22)
(-5,0.77)
(-4,1)
(-3,0.66)
(-2,0.66)
(-1,0.44)
};\label{functions-p1}
\addplot coordinates{
(-6,0)
(-5,0)
(-4,0)
(-3,0.11)
(-2,0.55)
(-1,1)
};\label{functions-p2}
\addplot[color=black,mark=triangle*] coordinates{
(-6,0)
(-5,0)
(-4,0)
(-3,0.55)
(-2,0.88)
(-1,0.66)
};\label{functions-p3}
\end{axis}
\node[align=center, text=black] at (rel axis cs:0.15,1.13) {($0$)};
\end{tikzpicture}
{\color{gray}\vrule}
\begin{tikzpicture}
\begin{axis}
[
legend style = { at = {(0.97,0.4)}},
xlabel = $\epsilon$ ($10^x$),
ymax = 1,
ymin =0,
width=0.22\textwidth
]
\addplot coordinates{
(-6,0)
(-5,0)
(-4,0.33)
(-3,1)
(-2,1)
(-1,0.67)
};\label{functions-p1}
\addplot coordinates{
(-6,0)
(-5,0)
(-4,0)
(-3,0)
(-2,0)
(-1,0)
};\label{functions-p2}
\addplot[color=black,mark=triangle*] coordinates{
(-6,0)
(-5,0)
(-4,0)
(-3,0.33)
(-2,1)
(-1,0.67)
};\label{functions-p3}
\end{axis}
\node[align=center, text=black] at (rel axis cs:0.24,1.13) {\textbf{Adult}};
\node[align=center, text=black] at (rel axis cs:0.52,1.13) {($0$)};
\end{tikzpicture}
\begin{tikzpicture}
\begin{axis}
[
legend style = { at = {(0.97,0.4)}},
xlabel = $\epsilon$ ($10^x$),
ymax = 1,
ymin =0,
width=0.22\textwidth
]
\addplot coordinates{
(-6,0)
(-5,0.33)
(-4,0.67)
(-3,1)
(-2,1)
(-1,0.67)
};\label{functions-p1}
\addplot coordinates{
(-6,0)
(-5,0)
(-4,0)
(-3,0)
(-2,0)
(-1,0.33)
};\label{functions-p2}
\addplot[color=black,mark=triangle*] coordinates{
(-6,0)
(-5,0)
(-4,0)
(-3,1)
(-2,1)
(-1,0.67)
};\label{functions-p3}
\end{axis}
\node[align=center, text=black] at (rel axis cs:0.15,1.13) {($0$)};
\end{tikzpicture}
{\color{gray}\vrule}
\begin{tikzpicture}
\begin{axis}
[
legend style = { at = {(0.97,0.4)}},
xlabel = $\epsilon$ ($10^x$),
ymax = 1,
ymin =0,
width=0.22\textwidth
]
\addplot coordinates{
(-6,0)
(-5,0.08)
(-4,0.69)
(-3,0.85)
(-2,0.85)
(-1,0)
};\label{functions-p1}
\addplot coordinates{
(-6,0)
(-5,0)
(-4,0)
(-3,0)
(-2,0)
(-1,0.08)
};\label{functions-p2}
\addplot[color=black,mark=triangle*] coordinates{
(-6,0)
(-5,0)
(-4,0)
(-3,0.31)
(-2,0.92)
(-1,0.92)
};\label{functions-p3}
\end{axis}
\node[align=center, text=black] at (rel axis cs:0.25,1.13) {\textbf{Flight}};
\node[align=center, text=black] at (rel axis cs:0.53,1.13) {($0$)};
\end{tikzpicture}
\begin{tikzpicture}
\begin{axis}
[
legend style = { at = {(0.97,0.4)}},
xlabel = $\epsilon$ ($10^x$),
ymax = 1,
ymin =0,
width=0.22\textwidth
]
\addplot coordinates{
(-6,0.08)
(-5,0.92)
(-4,1)
(-3,0.92)
(-2,0.85)
(-1,0)
};\label{functions-p1}
\addplot coordinates{
(-6,0.08)
(-5,0.08)
(-4,0.08)
(-3,0.08)
(-2,0.15)
(-1,0.92)
};\label{functions-p2}
\addplot[color=black,mark=triangle*] coordinates{
(-6,0.08)
(-5,0.08)
(-4,0.08)
(-3,0.92)
(-2,0.92)
(-1,0.92)
};\label{functions-p3}
\end{axis}
\node[align=center, text=black] at (rel axis cs:0.25,1.13) {($0.08$)};
\end{tikzpicture}
{\color{gray}\vrule}
\begin{tikzpicture}
\begin{axis}
[
legend style = { at = {(0.97,0.4)}},
xlabel = $\epsilon$ ($10^x$),
ymax = 1,
ymin =0,
width=0.22\textwidth
]
\addplot coordinates{
(-6,0.6)
(-5,0.6)
(-4,0.6)
(-3,0.6)
(-2,0.8)
(-1,0.8)
};\label{functions-p1}
\addplot coordinates{
(-6,0)
(-5,0)
(-4,0)
(-3,0.4)
(-2,0.6)
(-1,0.6)
};\label{functions-p2}
\addplot[color=black,mark=triangle*] coordinates{
(-6,0)
(-5,0)
(-4,0)
(-3,0.6)
(-2,1)
(-1,1)
};\label{functions-p3}
\end{axis}
\node[align=center, text=black] at (rel axis cs:0.23,1.13) {\textbf{Voter}};
\node[align=center, text=black] at (rel axis cs:0.5,1.13) {($0$)};
\end{tikzpicture}
\begin{tikzpicture}
\begin{axis}
[
legend style = { at = {(0.97,0.4)}},
xlabel = $\epsilon$ ($10^x$),
ymax = 1,
ymin =0,
width=0.22\textwidth
]
\addplot coordinates{
(-6,0.7)
(-5,0.7)
(-4,1)
(-3,1)
(-2,0.9)
(-1,0.9)
};\label{functions-p1}
\addplot coordinates{
(-6,0.7)
(-5,0.7)
(-4,0.7)
(-3,0.7)
(-2,0.7)
(-1,0.9)
};\label{functions-p2}
\addplot[color=black,mark=triangle*] coordinates{
(-6,0.7)
(-5,0.7)
(-4,0.8)
(-3,0.8)
(-2,0.8)
(-1,1)
};\label{functions-p3}
\end{axis}
\node[align=center, text=black] at (rel axis cs:0.15,1.13) {($0$)};
\end{tikzpicture}
}
\vspace{-1em}
\caption{G-recall for varying thresholds under $f_1$ (\ref{functions-p1}), $f_2$ (\ref{functions-p2}) and $f_3$ (\ref{functions-p3}) for spread (left) and skewed (right) noise.}
    \label{fig:app-fun}
\end{figure*}
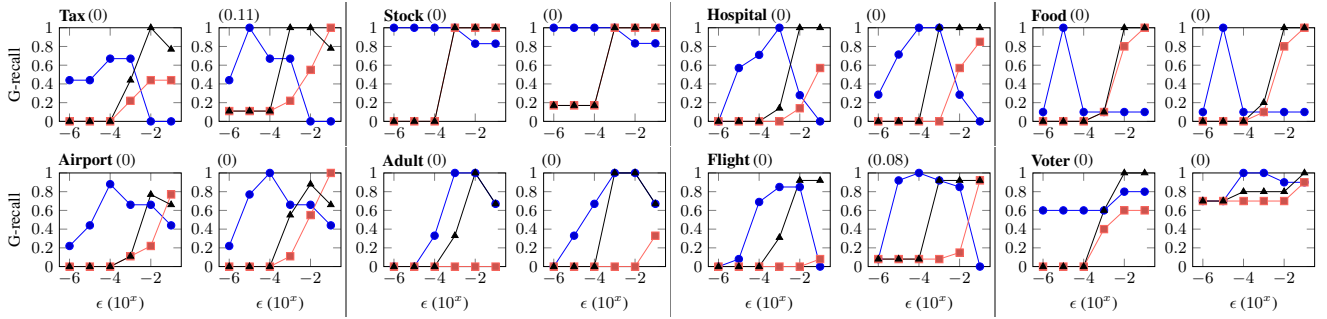

\begin{table*}[t]
\small
    \centering
    \scalebox{0.93}{
    \begin{tabular}{|c||c|}
\hline
\rowcolor{lightgray}\textbf{Approximate DC} & \textbf{Valid DC}\\\hline
$\forall t,t'\neg(t[\att{St}]=t'[\att{St}]\wedge t[\att{Salary}]> t'[\att{Salary}]\wedge t[\att{Tax}]< t'[\att{Tax}])$ & \makecell{$\forall t,t'\neg(t[\att{St}]=t'[\att{St}]\wedge t[\att{Salary}]> t'[\att{Salary}]\wedge t[\att{Tax}]< t'[\att{Tax}]\wedge t[\att{G}]= t'[\att{G}]$\\$\wedge t[\att{SE}]\ge t'[\att{SE}]\wedge t[\att{Ph}]= t'[\att{Ph}])$}\\\hline
$\forall t,t'\neg(t[\att{High}]< t[\att{Low}])$ & \makecell{$\forall t,t'\neg(t[\att{High}]< t[\att{Low}]\wedge t[\att{Open}]< t[\att{High}]\wedge t[\att{Low}]\le t[\att{Close}])$}\\\hline
$\forall t,t'\neg(t[\att{St}]=t'[\att{St}]\wedge t[\att{MC}]= t'[\att{MC}]\wedge t[\att{StAvg}]\neq t'[\att{StAvg}])$ & \makecell{$\forall t,t'\neg(t[\att{St}]=t'[\att{St}]\wedge t[\att{MC}]= t'[\att{MC}]\wedge t[\att{StAvg}]\neq t'[\att{StAvg}]$\\$\wedge t[\att{City}]= t'[\att{City}]\wedge t[\att{Sample}]= t'[\att{Sample}])$}\\\hline
$\forall t,t'\neg(t[\att{Zip}]=t'[\att{Zip}]\wedge t[\att{St}]\neq t'[\att{St}])$ & \makecell{$\forall t,t'\neg(t[\att{Zip}]=t'[\att{Zip}]\wedge t[\att{St}]\neq t'[\att{St}])\wedge t[\att{Name}]= t'[\att{Name}]$\\$\wedge t[\att{FacilityType}]\neq t'[\att{FacilityType}])$}\\\hline
\makecell{$\forall t,t'\neg(t[\att{OSt}]=t'[\att{OSt}]\wedge t[\att{DSt}]= t'[\att{DSt}]\wedge t[\att{DTime}]\ge t'[\att{DTime}])$\\$\wedge t[\att{ATime}]\le t'[\att{ATime}])\wedge t[\att{ETime}]> t'[\att{ETime}])$} & \makecell{$\forall t,t'\neg(t[\att{St}]=t'[\att{St}]\wedge t[\att{MC}]= t'[\att{MC}]\wedge t[\att{SA}]= t'[\att{SA}]\wedge t[\att{PN}]= t'[\att{PN}]$\\$\wedge t[\att{Owner}]\neq t'[\att{Owner}])$}\\\hline
$\forall t,t'\neg(t[\att{Age}]< t'[\att{Age}]\wedge t[\att{BirthYear}]< t'[\att{BirthYear}])$ & \makecell{$\forall t,t'\neg(t[\att{Age}]< t'[\att{Age}]\wedge t[\att{BirthYear}]< t'[\att{BirthYear}]\wedge t[\att{C}]\neq t'[\att{C}]$\\$\wedge t[\att{Status}]\neq t'[\att{Status}]\wedge t[\att{Reason}]= t'[\att{Reason}])$}\\
\hline
\end{tabular}}
\caption{Approximate vs Valid DCs. Attributes: $\att{St}$ -- state, $\att{Ph}$ -- phone, $\att{G}$ -- gender, $\att{SE}$ -- single exemption, $\att{MC}$ -- measure code, $\att{OSt}$, $\att{DSt}$ -- origin and destination state, $\att{Dtime}$, $\att{ATime}$, $\att{ETime}$ -- departure, arrival, and elapsed time, $\att{C}$ -- county. 
\label{table:ADCvsDC}}
\end{table*}

Another interesting phenomenon is that for we consistently obtain a higher G-recall on the error-concentrated datasets (especially for the functions $f_2$ and $f_3$). This is expected, especially for the function $f_3$, as when the errors are concentrated in a small subset of the tuples, these tuples will participate in every violation of the DC, and we only need to remove them from the database to satisfy the DC. The function $f_3$ (or, more accurately, our greedy approximation algorithm for this function) usually behaves better than the function $f_2$, especially on the error-concentrated datasets, and we are able to obtain a higher G-recall for a larger range of thresholds. As explained in Section~\ref{sec:functions}, this is due to the fact that one erroneous tuple may result in a set of problematic tuples that contains every tuple in the database, while if we just remove this tuple, the DC will be satisfied. \ester{For this same reason, while with the function $f_2$ we constantly obtain the best accuracy using $\epsilon=10^{-1}$, with the function $f_3$, we sometimes obtain better results with the smaller threshold $\epsilon=10^{-2}$.}

\ester{Observe that in the experiments reported in Figure~\ref{fig:sample}, we have used six specific thresholds, with which we are not always able to obtain the highest possible G-recall. If we conduct a more refined analysis, we find that using the threshold $5\times 10^{-5}$ for the function $f_1$ on the Tax dataset, for example, we are able to obtain a G-recall of $1$. When we increase the threshold, we are able to obtain more general DCs (i.e., DCs consisting of less predicates) that we cannot obtain using smaller thresholds; however, some DCs become ``too general'', and we also lose some of the good DCs that we obtained using the smaller threshold. Hence, we need to find the threshold that will generate the best results with high probability. Using the above insights, we can choose a certain threshold (that depends on the approximation function), that will generate good results with high probability. Based on Figure~\ref{fig:app-fun}, the best thresholds in that sense are $10^{-4}$, $10^{-2}$, and $10^{-1}$ for the functions $f_1$, $f_2$, and $f_3$, respectively. Using these thresholds we obtained an average G-recall of $0.71$, $0.72$, and $0.97$, respectively.}

\ester{Finally, Table~\ref{table:ADCvsDC} presents some of the golden DCs that we were able to obtain with the three approximation functions using the best threshold according to Figure~\ref{fig:app-fun}, as well as an example of a corresponding valid DC from the same dirty dataset, obtained with the threshold $\epsilon=0$. The DCs were obtained from the Tax, SP Stock, Hospital, Food, Flight, and NCVoter datasets}. Many valid DCs are obtained from a single approximate DC by adding more predicates to cover for the errors in the database, which results in longer and less general DCs. Therefore, we often obtain less DCs and shorter DCs when considering ADCs. However, this is not always the case, as in some cases we also discover constraints that are approximate DCs, but cannot be extended to any minimal valid DC.

For example, the DC stating that the same zip code cannot correspond to two states (obtained from the Food dataset) becomes the DC stating that the same zip code cannot correspond to two states if the name and the type of the facility are the same. Clearly, we do not expect to obtain such complicated constraints, which strengthens our motivation for considering ADCs rather than valid DCs. In fact, while this DC generally holds, there are a few multi-state US zip codes (e.g., the zip code 84536 belongs to Utah and Arizona). If our original database contained two tuples with the same zip code and different states we would not be able to discover this DC unless considering ADCs. This example shows that ADCs are meaningful even when the database is clean, as they allow us to discover rules that are generally correct, but may have a few exceptions (about $0.03\%$ of zip codes in the US cross states).

%% file: conclusions.tex
\section{Concluding Remarks}\label{sec:conclusions}
We investigated the problem of detecting and enumerating minimal ADCs from data. We introduced a formal definition of an ADC based on a general family of approximation functions that subsumes previous proposals. We devised an algorithm for enumerating minimal ADCs and experimentally evaluated its performance on both real-world and synthetic datasets. Our experimental results showed that constructing the input to our enumeration algorithm requires orders of magnitude more time than enumerating the ADCs for large datasets. We showed that we are able to obtain good results (with high precision and recall) from a sample while avoiding the high computational cost. We also provided a theoretical analysis for the problem of discovering ADCs from a sample.

The computational complexity of enumerating DCs remains an open problem (in terms of combined complexity, where both the schema and database are given as input). In particular, it would be interesting to understand whether this problem is equivalent, harder or easier than the minimal hitting set problem. The main difference between the two problems is our knowledge about the relationships between the elements (e.g., we know that if a tuple pair does not satisfy a predicate then it satisfies the complement predicate). It is not clear how this additional information affects the complexity. 
